\documentclass[journal]{IEEEtran}
\usepackage{placeins}
\usepackage{diagbox}
\usepackage[inline]{enumitem}
\usepackage{xcolor}
\usepackage{amsmath,amsfonts,amssymb}
\usepackage{algorithm}
\usepackage[group-separator={,},group-minimum-digits=4]{siunitx}
\usepackage{threeparttable}
\usepackage{array}
\usepackage[caption=false,font=normalsize,labelfont=sf,textfont=sf]{subfig}
\usepackage{textcomp}
\usepackage{stfloats}
\usepackage{url}
\usepackage{verbatim}
\usepackage{graphicx}
\usepackage{algpseudocode}
\usepackage{cite}
\usepackage{balance}
\usepackage{mathrsfs}
\usepackage[colorlinks=true,
            linkcolor=blue,
            anchorcolor=blue,
            citecolor=blue]{hyperref}
\newtheorem{theorem}{Theorem}

\AtBeginDocument{%
  \setlength{\abovedisplayskip}{4pt plus 1pt minus 1pt}%
  \setlength{\belowdisplayskip}{4pt plus 1pt minus 1pt}%
  \setlength{\abovedisplayshortskip}{3pt plus 1pt}%
  \setlength{\belowdisplayshortskip}{3pt plus 1pt minus 1pt}%
}
\begin{document}

\title{ 
\fontsize{17.5pt}{\baselineskip}\selectfont
{Cyclic-Prefix OFDM Probing for Spatial-ISI-Free Distributed Acoustic Sensing via Frequency-Domain Channel Reconstruction}
} 
\author{ 
{
    Huan~Huang,~\textit{Member,~IEEE},
    Zhiyang~Xue,
    Ziang~Chen,
    Zhongxing~Tian,
    Dongdong~Zou,
    Gangxiang~Shen,~\textit{Senior~Member,~IEEE,~Fellow,~Optica},
    and Yi~Cai,~\textit{Fellow,~Optica}
}
\thanks{ 

H.~Huang, Z.~Xue, Z.~Chen, Z.~Tian, D.~Zou, G.~Shen, and Y.~Cai are with the School of Electronic and Information Engineering, Soochow University, Suzhou, Jiangsu 215006, China 
(e-mail: hhuang1799@gmail.com; zyxue999@stu.suda.edu.cn; zachen@stu.suda.edu.cn; ddzou@suda.edu.cn; zxtian@ieee.org; shengx@suda.edu.cn; yicai@suda.edu.cn).   
}
}
\maketitle

\begin{abstract}\hbadness=7000\relax
In matched-filter-based pulse-compression distributed acoustic sensing (DAS), nonzero compression sidelobes cause deterministic inter-range-bin leakage, known as spatial inter-symbol interference (ISI), and may produce false responses in reconstructed Rayleigh-backscatter traces.
In this work, we propose a cyclic-prefix orthogonal frequency-division multiplexing (CP-OFDM) DAS system for phase-sensitive optical time-domain reflectometry ($\phi$-OTDR) that uses a data-bearing CP-OFDM waveform directly as the sensing probe. The same waveform recovers forward communication data, providing an initial demonstration of shared-waveform integrated sensing and communication (ISAC). 
To the best of our knowledge, this is the first time that distributed Rayleigh backscattering has been formulated as a finite-memory sensing multipath channel. Based on this formulation, we rigorously prove that, if the useful OFDM and CP lengths cover the sensing multipath channel memory, then CP removal, one-tap frequency-domain equalization, and inverse discrete Fourier transform reconstruct each range-bin coefficient without deterministic waveform-induced spatial ISI for subsequent gauge-differential phase demodulation. 
For a simulated 5.2-km link containing ten simultaneous strong and weak events with intragroup spacings of 5.31--5.83 m, the proposed CP-OFDM-aided frequency-domain receiver suppresses spatial-ISI-induced off-event leakage and improves the phase-trace mean-square error by a maximum of 29.55~dB over a matched-filter-based pulse-compression receiver. 
In a heterodyne coherent experiment over a 5.2-km fiber link with a 111.984-MHz occupied bandwidth, CP-OFDM recovers 500-Hz vibrations near 80 m under 1- and 5-V piezoelectric transducer (PZT) drives. Under the same probing conditions, CP-OFDM improves the vibration-to-background ratio by approximately 3.2 and 6.2~dB, respectively, compared with matched-filter-based linear frequency-modulated (LFM) pulse compression, while preserving the local vibration waveform.
The same data-bearing probe also recovers an image with zero measured bit-error rate and a median error vector magnitude of -23.14~dB. 
These results validate the CP-OFDM-aided frequency-domain channel reconstruction for spatial-ISI-free DAS and demonstrate its potential for shared-waveform optical-fiber ISAC.
\end{abstract}


\section{Introduction}\label{Intro}
\IEEEPARstart{D}{istributed} acoustic sensing (DAS) transforms an optical fiber into a dense acoustic aperture by detecting changes induced by perturbations along the propagation path. Beyond the general principles and system trade-offs summarized in~\cite{HartogDOFSBook,LuDOFSReview}, recent dark-fiber demonstrations have shown its practical value. The same installed fiber infrastructure can observe ocean and earthquake wave fields, image earthquake rupture processes, and support city-scale sensing~\cite{LindseyScience2019,LiNatureDAS2023,LiuNatCommUrban2025}. These applications require a long sensing range, high sensitivity, and reliable spatial discrimination when multiple acoustic events occur along the fiber.

Existing DAS systems can be generally divided into two categories: forward-transmission/interferometric and backward-reflectometric.
The forward route detects perturbations carried by the transmitted optical field. In contrast, the backward route resolves Rayleigh backscattering from one fiber end.
The single-end-access capability of the backward route is particularly attractive for field-deployed fiber links, where access to the far end may be unavailable or impractical. Therefore, this paper focuses on backward Rayleigh-backscatter DAS. 
In backward DAS, phase-sensitive optical time-domain reflectometry ($\phi$-OTDR) maps the time of flight of a pulse to position and extracts vibration information from coherent Rayleigh backscattering.
Short-pulse $\phi$-OTDR established the basic framework for intrusion and vibration sensing, as described in~\cite{JuarezJLT2005,LuJLT2010}. Since then, field and long-haul studies have advanced this approach to practical applications and extended spans through long-perimeter deployment, Raman amplification, Brillouin amplification, and high-sensitivity pipeline monitoring, as detailed in~\cite{JuarezAO2007,MartinsJLT2014,WangOL2014,PengOE2014}.

Backward DAS also includes optical frequency-domain reflectometry (OFDR) routes. Rayleigh-scatter OFDR and frequency-modulated continuous-wave reflectometry provide distributed measurements with high spatial resolution by resolving frequency-to-delay information instead of launching isolated short pulses~\cite{FroggattAO1998,VenkateshAO1990}. 
More recent $\Phi$-OFDR and optical frequency-domain DAS systems have demonstrated high-frequency or crosstalk-suppressed acoustic sensing~\cite{MarconOE2019,LiOL2020}, which show that backscatter can be treated as a distributed channel in the frequency domain. 
However, their sweep linearity, range, and update-rate constraints differ from those of the pulsed $\phi$-OTDR. 
For pulsed $\phi$-OTDR, spatial resolution and probe energy exhibit a fundamental trade-off. While a shorter pulse improves spatial resolution, it reduces probe energy under peak-power and nonlinear-effect constraints.

Broadband pulse-compression probing alleviates this trade-off by using a long-duration waveform to increase the probe energy and a large bandwidth to retain fine spatial resolution, followed by matched-filter-based pulse compression at the receiver.
Chirped-pulse $\phi$-OTDR enabled single-shot distributed temperature and strain tracking~\cite{PastorOE2016}, while optical pulse-compression reflectometry formalized the use of long coded or chirped probes for high-resolution reflectometry~\cite{ZouOE2015}. Subsequent developments have addressed practical limitations of this approach, including phase-noise compensation in pulse-compression DAS~\cite{PineiroJLT2023} and time-expanded $\phi$-OTDR for high spatial resolution with reduced detection bandwidth~\cite{SorianoLSA2021}.

Coding, multiplexing, and diversity have further improved the robustness of $\phi$-OTDR. 
Perfect periodic correlation codes and bipolar coding exploit the correlation structure to enhance range, noise performance, fading compensation, or frequency-drift tolerance~\cite{SaguesOE2021,WuJLT2020}. Coherent multiple-input multiple-output (MIMO) sensing and fully digital MIMO orthogonal frequency-division multiplexing (MIMO-OFDM) $\Delta\phi$-OTDR introduce polarization/frequency diversity to mitigate fading~\cite{GuerrierOE2020,GuerrierOE2021}. 
Recent multi-layer, multi-frequency, and polyphase-code receivers continue this trend by suppressing fading or crosstalk through structured multiplexing, frequency diversity, or correlation processing~\cite{TLSMOTDR,WakisakaJLT2025,DengJLT2026}. 
Taken together, the literature has shifted progressively from simple, short pulses to waveforms with a larger time-bandwidth product and a richer coding structure.

However, most broadband-pulse and coded-probe DAS receivers still rely on matched-filter-based pulse compression. In these receivers, the recovered spatial trace is the Rayleigh-backscatter channel blurred by the autocorrelation response of the probing waveform, rather than the channel itself. Pulse-compression theory clarifies the underlying reason: a practical finite-duration and bandwidth-limited waveform has a finite compression response with a main lobe and sidelobes~\cite{LevanonRadarSignals,RichardsRadar}. In optical pulse-compression reflectometry, sidelobe apodization can reduce leakage but broadens the main lobe, thereby degrading spatial resolution~\cite{MompoOL2018}. Positive/negative swept pulses can extend the effective swept bandwidth, while asymmetric correlator receivers can suppress coding crosstalk; both still operate within a pulse-compression framework~\cite{XiaoJLT2024,DengJLT2026}. Consequently, residual sidelobe leakage appears as waveform-induced spatial inter-symbol interference (ISI), especially when a strong perturbation is located near a weak event.

This observation motivates a different DAS reconstruction paradigm: rather than accepting an autocorrelation-blurred matched-filter output, the receiver should directly reconstruct the distributed Rayleigh-backscatter channel over the sensing window.
To this end, we formulate the fast-time distributed Rayleigh-backscatter response as a finite-memory sensing multipath channel and propose a data-bearing cyclic-prefix orthogonal frequency-division multiplexing (CP-OFDM) DAS system for $\phi$-OTDR.
With a sufficient cyclic prefix, the linear backscatter convolution becomes circular convolution, enabling one-tap frequency-domain equalization and inverse discrete Fourier transform (IDFT)-based reconstruction of the range-bin coefficients without deterministic waveform-induced spatial ISI. At a forward communication receiver, the same waveform carries a recoverable data payload, providing an initial demonstration of shared-waveform integrated sensing and communication (ISAC).
The main contributions are summarized as follows:
\begin{itemize}[leftmargin=*,topsep=2pt,itemsep=2pt,parsep=0pt]
\item
We establish a unified discrete-time model for pulse-compression $\phi$-OTDR DAS and identify the mechanism of waveform-induced spatial ISI. Starting from the Rayleigh-backscatter reception model, we represent the sampled return as the linear convolution between a generic probing waveform and the distributed fiber response. 
We demonstrate that matched-filter-based pulse compression reconstructs the Rayleigh-backscatter channel convolved with the probing-waveform autocorrelation, rather than the channel itself. 
By separating the desired range-bin response from the sidelobe-induced leakage and propagating both terms through gauge-differential phase recovery, the analysis explains how deterministic inter-range-bin coupling produces amplitude bias, phase distortion, and false spatial responses, particularly when weak perturbations are located near strong events. This model also distinguishes waveform-induced spatial ISI from receiver noise, Rayleigh fading, and gauge overlap.

\item
To the best of our knowledge, this is the first work to model distributed Rayleigh backscattering as a finite-memory sensing multipath channel and process it through a CP-OFDM channel-reconstruction framework. In this formulation, fast time resolves the propagation delays of distributed range bins, while slow time describes perturbation-induced channel evolution.
Our formulation uses the cyclic prefix to convert the finite-memory linear convolution into circular convolution over the useful OFDM block.
When the useful OFDM and CP lengths cover the sensing multipath channel memory, known nonzero subcarrier symbols enable one-tap frequency-domain equalization. An inverse discrete Fourier transform then reconstructs each spatial channel tap for gauge-differential phase demodulation. The proposed receiver therefore estimates the distributed Rayleigh-backscatter channel itself rather than its waveform-autocorrelation-blurred version. Because subcarrier symbols can carry payload data, the same probe can support sensing-channel reconstruction and forward communication simultaneously.

\item
We rigorously establish the spatial-ISI-free property of the proposed CP-OFDM-aided frequency-domain channel reconstruction and clarify its spatial-resolution boundary. For a channel that is quasi-static within one OFDM probing block, a useful OFDM length no shorter than the channel memory and a cyclic prefix no shorter than the sensing multipath channel memory yield a reconstructed range-bin coefficient equal to the desired channel tap plus reconstructed noise, with no deterministic sidelobe leakage from other range bins. 
We further compare linear frequency-modulated (LFM) probing with matched-filter-based pulse compression and CP-OFDM probing with frequency-domain channel reconstruction under the same occupied bandwidth. Both have the bandwidth-limited spatial resolving interval $v_g/(2B_{\rm occ})$. A grid-aligned scatterer has a Kronecker-delta reconstruction response, whereas an off-grid scatterer exhibits the finite-bandwidth Dirichlet response. The advantage is therefore the elimination of pulse-compression sidelobe coupling on the range-bin grid while preserving the physical resolution imposed by the occupied bandwidth.

\item
Simulations and experiments verify the analysis under dense-event and practical link conditions. In a 5.2-km fiber simulation with a 128-MHz bandwidth, ten simultaneous strong and weak events are arranged in three groups, with an intragroup spacing of 5.31--5.83 m. Compared with the matched-filter-based LFM receiver, the proposed CP-OFDM-aided frequency-domain receiver suppressed the off-event spatial floor and improved the phase-trace mean-square error by 1.94--29.55~dB for all ten events. 
A heterodyne coherent experiment over a 5.2-km fiber link with a 111.984-MHz occupied bandwidth recovers 500-Hz vibrations near 80 m under 1- and 5-V piezoelectric transducer (PZT) drives, with correlations of 0.9975 and 0.9995 to fitted sinusoids, respectively. Under the same occupied bandwidth, emitted pulse duration, repetition period, and launch power, CP-OFDM improves the vibration-to-background ratio (VBR) by approximately 3.2 and 6.2~dB, respectively, compared to matched-filter-based LFM pulse compression at the two drive levels. Using the same data-bearing probe, an image is recovered with a zero measured bit-error rate and a median error vector magnitude of \(-23.14\)~dB, providing an initial shared-waveform ISAC validation.
\end{itemize}

The rest of this paper is organized as follows. Section~\ref{Princ} models the distributed Rayleigh-backscatter channel and analyzes waveform-induced spatial ISI in matched-filter-based DAS. Section~\ref{CPOFDM} presents the CP-OFDM frequency-domain channel-reconstruction receiver, proves its spatial-ISI-free property, and analyzes spatial resolution. Sections~\ref{Simu} and~\ref{Exp} present the simulation and experimental results, respectively. Finally, Section~\ref{Conclu} concludes this paper.


\section{Principles}\label{Princ}
In Section~\ref{SysMod}, we formulate the phase-sensitive optical time-domain reflectometry-based DAS system and derive the Rayleigh-backscatter reception model, where the distributed fiber response is represented as a finite-memory sensing multipath channel. 
Section~\ref{ISIByMF} describes acoustic-phase recovery for the conventional matched-filter-based pulse-compression receiver and quantifies the waveform-induced spatial ISI and its impact on the gauge-differential phase estimate.
\subsection{Phase-Sensitive Optical Time-Domain Reflectometry}\label{SysMod}

\begin{figure}[!t]
    \centering
    \includegraphics[width=0.45\textwidth]{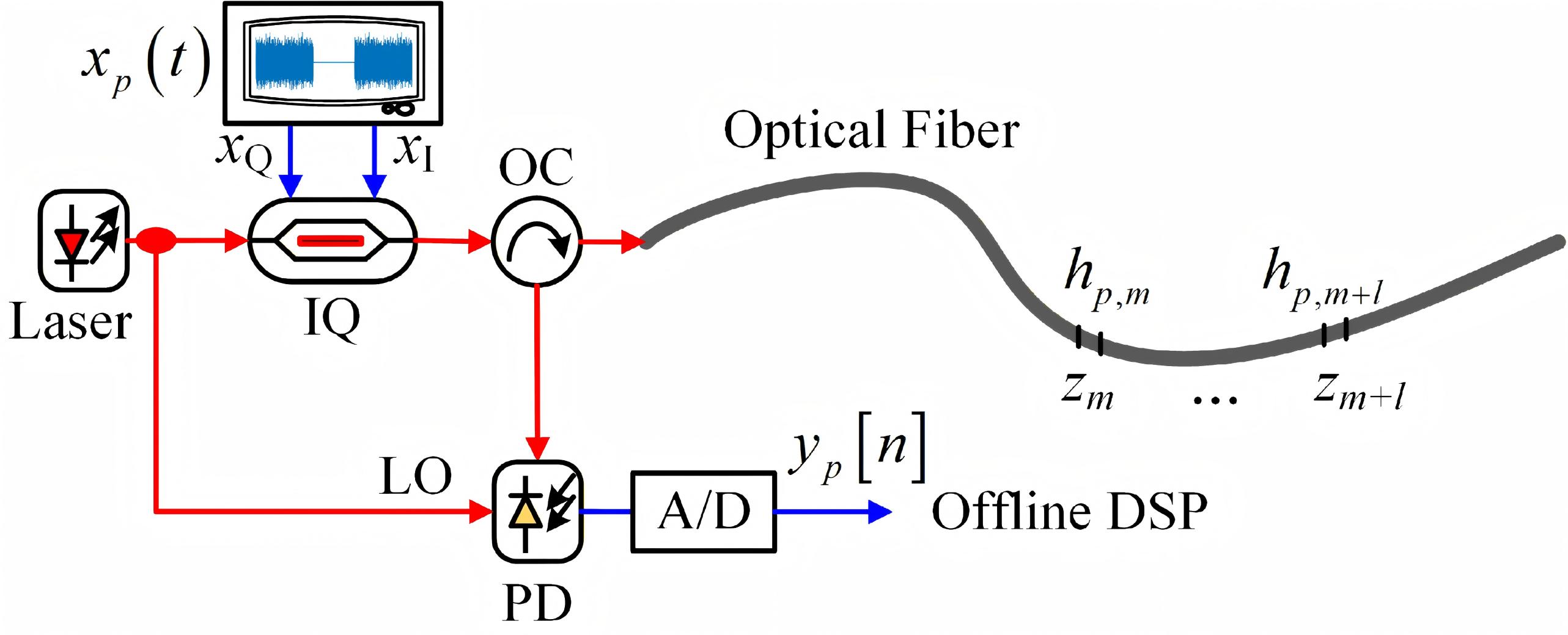}
    \caption{A DAS system based on phase-sensitive optical time-domain reflectometry ($\phi$-OTDR). IQ: in-phase/quadrature modulator; OC: optical circulator; PD: photodetector; A/D: analog-to-digital converter; DSP: digital signal processing.}
    \label{fig1}
\end{figure}

Fig.~\ref{fig1} illustrates a $\phi$-OTDR-based DAS system. The optical field from a narrow-linewidth laser is modulated by an in-phase/quadrature (IQ) modulator driven by a complex baseband probing waveform. 
Let \(x_p(t)\) denote the complex baseband probing waveform transmitted in the \(p\)-th probing period, where \(p\) is the slow-time index. In this subsection, \(x_p(t)\) is kept generic: short-pulse probes represent the standard $\phi$-OTDR intrusion and vibration implementations~\cite{JuarezJLT2005,LuJLT2010}; LFM probes represent pulse-compression OTDR-type probing~\cite{PastorOE2016,ZouOE2015}; and digitally coded probes cover periodic-code and polyphase-code DAS schemes~\cite{SaguesOE2021,DengJLT2026}. 

In the complex-envelope representation, the optical field launched in the \(p\)-th probing period is
\begin{equation}
    E_{{\rm tx},p}(t)=\sqrt{P_{\rm tx}}\,x_p(t),
    \label{eq:tx_field}
\end{equation}
where \(P_{\rm tx}\) is the launched optical power. The modulated optical signal is launched into the fiber through the optical circulator (OC). The Rayleigh-backscatter field is then coherently mixed with a local oscillator (LO) by using a photodetector (PD), sampled by the analog-to-digital (A/D) converter, and processed offline.

The backscattered field is formed through the coherent superposition of Rayleigh scattering from all infinitesimal fiber segments. Under the single-scattering approximation, the received complex baseband signal can be modeled as~\cite{LuJLT2010,MasoudiOE2017}
\begin{equation}
    y_p(t)
    =\int_{0}^{L}h_p(z) x_p\!\left(t-\tau(z)\right)dz + w_p(t),
    \label{eq:cont_backscatter}
\end{equation}
where \(L\) is the fiber length, \(w_p(t)\) denotes the equivalent receiver noise, $\tau(z)= {2z}/{v_g}$ is the round-trip propagation delay from the fiber input to position \(z\), and $v_g={c}/{n_g}$ is the group velocity in the fiber with \(c\) and \(n_g\) denoting the speed of light in vacuum and the group index, respectively. Detection-noise sources and their impact in direct-detection and coherent $\phi$-OTDR systems have been characterized in~\cite{LuOE2021Noise,VidalOE2023Noise}.  

In~\eqref{eq:cont_backscatter}, the distributed channel response \(h_p(z)\) denotes the equivalent complex Rayleigh-backscatter coefficient at position \(z\) during the \(p\)-th probing period. It accounts for the random scattering response, static round-trip propagation effects, fiber attenuation, and perturbation-induced phase variation. A compact representation is given by~\cite{MasoudiOE2017}
\begin{equation}
    h_p(z)=\rho(z)e^{-2\alpha z}e^{-j2\beta_0 z}e^{j\phi_p(z)},
    \label{eq:continuous_channel}
\end{equation}
where \(\rho(z)\) is the random complex Rayleigh scattering coefficient, \(\alpha\) is the amplitude attenuation coefficient of the fiber, \(\beta_0\) is the propagation constant at the optical carrier frequency, and \(\phi_p(z)\) denotes the vibration-induced phase perturbation at slow time \(p\). The constant optical and electrical gains of the coherent receiver are absorbed into \(h_p(z)\) and \(w_p(t)\).

After sampling the received signal at a sampling interval $T_s= {1}/{F_s}$, the standard $\phi$-OTDR time-of-flight mapping gives the spatial sampling interval
\begin{equation}
    \Delta z = \frac{v_gT_s}{2}.
    \label{eq:dz}
\end{equation}
Here, \(\Delta z\) in \eqref{eq:dz} is the spatial sampling interval of the discretized range bins, not necessarily the final spatial resolution of the DAS system. The achievable spatial resolution is also determined by the probing waveform bandwidth, the reconstruction or pulse-compression response, and the gauge operation used in the subsequent digital signal processing (DSP). 
Consequently, the \(m\)-th discrete range bin is located at
\begin{equation}
    z_m=m\Delta z,
    \label{eq:range_bin}
\end{equation}
where $m=0,1,\ldots,M-1$ and \(M \approx L/\Delta z\) is the number of discrete range bins. According to~\eqref{eq:dz}, the round-trip delay of the \(m\)-th bin satisfies
\begin{equation}
    \tau_m
    = \frac{2z_m}{v_g} = mT_s.
    \label{eq:integer_delay}
\end{equation}
Therefore, on this sampling grid, the contribution from the \(m\)-th range bin is a delayed version of the transmitted probing sequence.

Define $x_p[n]\triangleq x_p(nT_s)$ and $y_p[n]\triangleq y_p(nT_s)$,
and let \(h_{p,m}\) denote the aggregate complex Rayleigh-backscatter response within the \(m\)-th range bin. The sampled received signal can then be represented as the finite-memory linear convolution
\begin{equation}
    y_p[n]=\sum_{m=0}^{M-1} h_{p,m} x_p[n-m] + w_p[n],
    \label{eq:discrete_backscatter}
\end{equation}
where \(w_p[n]\triangleq w_p(nT_s)\). {Note that the complex Rayleigh-backscatter response is assumed to change sufficiently
little during the \(p\)-th probing period, which is a common assumption in $\phi$-OTDR-based DAS receivers~\cite{LuJLT2010,MasoudiOE2017}.}

Eq.~\eqref{eq:discrete_backscatter} shows that the distributed Rayleigh-backscatter response can be viewed as a finite-memory multipath channel along the fiber. We refer to this finite-memory channel as the ``sensing multipath channel.'' The fast-time index \(n\) resolves propagation delay, or equivalently distance, whereas the slow-time index \(p\) describes the temporal evolution of the scattering channel induced by external acoustic perturbations.
The following subsections use \eqref{eq:discrete_backscatter} to compare the conventional matched-filter-based pulse-compression receiver with the proposed CP-OFDM-aided frequency-domain channel-reconstruction receiver.

\subsection{Spatial ISI in Matched-Filter-Based DAS}\label{ISIByMF}
Based on the discrete backscatter model in \eqref{eq:discrete_backscatter}, the conventional pulse-compression receiver applies matched filtering to the received sequence using the transmitted probing sequence as the reference, as in radar pulse compression and optical pulse-compression reflectometry~\cite{RichardsRadar,MompoOL2018}. 
For notational simplicity, we assume that the same probing sequence is transmitted in each probing period, i.e., \(x_p[n]=s[n]\). The received signal in the \(p\)-th probing period then becomes
\begin{equation}
    y_p[n]
    =
    \sum_{m=0}^{M-1}
    h_{p,m}s[n-m]
    +
    w_p[n].
    \label{eq:mf_rx_model}
\end{equation}

The matched-filter output at range-bin index \(\ell\) is
\begin{equation}
    r_p[\ell]
    =
    \sum_n
    y_p[n]s^H[n-\ell].
    \label{eq:mf_output_def}
\end{equation}
Substituting \eqref{eq:mf_rx_model} into \eqref{eq:mf_output_def} gives
\begin{align}
    r_p[\ell]
    &=
    \sum_{m=0}^{M-1}
    h_{p,m}
    \sum_n
    s[n-m]s^H[n-\ell]
    +
    v_p[\ell]       \notag \\
    &=
    \sum_{m=0}^{M-1}
    h_{p,m}
    R_s[\ell-m]
    +
    v_p[\ell],
    \label{eq:mf_output_conv}
\end{align}
where $R_s[q]=\sum\nolimits_n s[n]s^H[n-q]$
is the aperiodic autocorrelation function of the probing sequence, and $v_p[\ell]=\sum\nolimits_n w_p[n]s^H[n-\ell]$ is the filtered noise term. 
Therefore, the matched-filter output is not the distributed Rayleigh-backscatter channel itself; rather, it is the convolution of the channel with the autocorrelation response of the transmitted probing sequence.

Let
\begin{equation}
    E_s = R_s[0] = \sum_n |s[n]|^2
    \label{eq:waveform_energy}
\end{equation}
denote the probing waveform energy, and define the normalized autocorrelation response as $c_s[q]= {R_s[q]}/{E_s}$.
After energy normalization, the matched-filter-based channel estimation is
\begin{align}
    \hat{h}^{\rm MF}_{p,\ell}
    &=
    \frac{r_p[\ell]}{E_s}  \notag \\
    &=
    h_{p,\ell}
    +
    \underbrace{\sum\limits_{\substack{m=0, m\ne \ell}}^{M-1} h_{p,m}c_s[\ell-m]}_{\text{spatial ISI}}
    +
    \tilde{v}_p[\ell],
    \label{eq:mf_channel_est}
\end{align}
where \(\tilde{v}_p[\ell]=v_p[\ell]/E_s\). The first term in~\eqref{eq:mf_channel_est} is the desired response of the \(\ell\)-th range bin. The second term originates from the off-peak autocorrelation response of \(c_s[q]\) and represents leakage from other range bins into the \(\ell\)-th bin. This waveform-induced spatial interference is hereafter termed spatial ISI because it appears as probing-waveform-dependent coupling among spatial range bins.

Ideally, if the normalized autocorrelation were a Kronecker delta, i.e., $c_s[q]=\delta[q]$, 
then \eqref{eq:mf_channel_est} would reduce to
\begin{equation}
    \hat{h}^{\rm MF}_{p,\ell}
    =
    h_{p,\ell}
    +
    \tilde{v}_p[\ell],
    \label{eq:ideal_mf_channel}
\end{equation}
and waveform-induced spatial ISI would be absent. 

However, any finite-duration, bandwidth-limited probing waveform that can be realized in practice exhibits a nonzero off-peak pulse-compression response, which is a standard consequence of pulse-compression waveform design~\cite{LevanonRadarSignals,RichardsRadar}. For an LFM pulse, the main-lobe width is approximately inversely proportional to the swept bandwidth, whereas the sidelobe level depends on the time- or frequency-domain weighting. In optical pulse-compression reflectometry, the authors in~\cite{MompoOL2018} showed that windowing suppresses sidelobes, but it also broadens the main lobe, thereby degrading spatial resolution. Therefore, matched filtering typically yields a finite spatial point spread function rather than a delta response.

The impact of the spatial ISI term in \eqref{eq:mf_channel_est} is particularly severe in dense-event DAS scenarios. At the location of a weak perturbation, leakage from a nearby strong perturbation can be comparable to, or even greater than, the desired weak-event response. More specifically, for a weak event at bin \(\ell\) and a strong event at bin \(m\), the interference component $h_{p,m}c_s[\ell-m]$
may distort the recovered phase at bin \(\ell\), even when \(|c_s[\ell-m]|\ll 1\), because the channel magnitude \(|h_{p,m}|\) or the perturbation-induced slow-time phase variation at bin \(m\) can be much larger than that of the weak event.

In practical $\phi$-OTDR-based DAS processing, the acoustic phase is typically extracted through a gauge-differential operation between two range bins separated by a gauge length~\cite{GabaiOL2016}. This gauge length is closely related to the spatial sensitivity of DAS. Let \(G\) denote the gauge length in samples. For the full-band, single-view case considered in this analysis, the complex gauge product of the matched-filter output is
\begin{equation}
    \hat{g}^{\rm MF}_{p,\ell}
    =
    \hat{h}^{\rm MF}_{p,\ell+G}
    \left(\hat{h}^{\rm MF}_{p,\ell}\right)^{\! H}.
    \label{eq:mf_gauge_product}
\end{equation}

To examine how the autocorrelation sidelobes affect phase recovery, we rewrite the matched-filter-based channel estimate in~\eqref{eq:mf_channel_est} as
\begin{equation}
    \hat{h}^{\rm MF}_{p,\ell}=h_{p,\ell}+\epsilon_{p,\ell},
    \label{eq:mf_error_def}
\end{equation}
where
\begin{equation}
    \epsilon_{p,\ell}
    =
    \sum_{\substack{m=0, m\neq \ell}}^{M-1}
    h_{p,m}c_s[\ell-m]
    +
    \tilde{v}_p[\ell]
    \label{eq:mf_error}
\end{equation}
collects the spatial-ISI leakage and filtered noise. Substituting \eqref{eq:mf_error_def} into \eqref{eq:mf_gauge_product} yields
\begin{align}
    \hat{g}^{\rm MF}_{p,\ell}
    &=
    h_{p,\ell+G}h^H_{p,\ell}
    +
    h_{p,\ell+G}\epsilon^H_{p,\ell}
    +
    \epsilon_{p,\ell+G}h^H_{p,\ell}
    +
    \epsilon_{p,\ell+G}\epsilon^H_{p,\ell}.
    \label{eq:mf_gauge_expansion}
\end{align}
The first term in \eqref{eq:mf_gauge_expansion} is the desired gauge response, whereas the remaining terms arise from matched-filter leakage and noise. Because the phase operator in \eqref{eq:common_phase_recovery} is nonlinear, these leakage terms bias the gauge-product amplitude and distort the recovered slow-time differential phase trace. Thus, matched-filter-based DAS channel estimation produces weak-event distortion or false spatial responses when multiple perturbations fall within the effective support of the off-peak pulse-compression response.

Increasing the receiver sampling rate alone does not remove this impairment. For an occupied bandwidth \(B_{\rm occ}\), if \(F_s>B_{\rm occ}\), the spatial sampling interval \(\Delta z=v_g/(2F_s)\) decreases, and the pulse-compression response \(R_s[q]\) is sampled more finely. However, the spatial main-lobe width remains approximately \(v_g/(2B_{\rm occ})\), and the off-peak response remains present. Oversampling can improve peak localization and fractional-delay estimation, but it does not change the matched-filter relation in~\eqref{eq:mf_output_conv}. Therefore, the waveform-induced spatial ISI in~\eqref{eq:mf_channel_est} remains an inherent limitation of matched-filter-based pulse-compression receivers.

\section{Cyclic-Prefix OFDM Probing and Channel Reconstruction for Spatial-ISI-Free DAS}
\label{CPOFDM}
In Section~\ref{CPOFDMPrinciple}, we present the proposed CP-OFDM probing scheme, together with the corresponding frequency-domain channel reconstruction and gauge-differential phase-recovery procedure. 
In Section~\ref{CPOFDMTheorem}, we prove that, under a sufficient-CP condition, the proposed receiver reconstructs the distributed Rayleigh-backscatter channel without deterministic waveform-induced spatial ISI.
In Section~\ref{SpatialResolu}, we analyze the spatial resolution of the proposed CP-OFDM probing, including its bandwidth-limited resolving interval.

\subsection{CP-OFDM Probing and Phase Recovery}\label{CPOFDMPrinciple}
Fig.~\ref{fig:cpofdm_dsp} illustrates the proposed CP-OFDM probing and the frequency-domain channel-reconstruction procedure for matched-filter-free DAS.
As established in Section~\ref{SysMod}, the distributed Rayleigh-backscatter
response in~\eqref{eq:discrete_backscatter} can be treated as the sensing multipath channel.
The finite-memory linear convolution in~\eqref{eq:discrete_backscatter} has the same form as a finite-memory multipath-channel convolution~\cite{ShiehOFDMBook,GoldsmithBook}. 
Let \(N\) denote the length of one useful OFDM probing block, and let $\mathcal{K} = \{0,1,\ldots,N-1\}$ be the set of active subcarrier indices in the critically sampled full-band case. The frequency-domain probing weights are denoted by \(S[k]\), \(k\in\mathcal{K}\). 
The receiver knows \(S[k]\) and $S[k]\neq 0$ for $k=0,1,\ldots,N-1$.
In practice, constant-modulus weights, such as quadrature phase-shift keying (QPSK) symbols, are preferred to avoid noise enhancement during frequency-domain channel reconstruction.

\begin{figure}[!t]
    \centering
    \includegraphics[width=0.486\textwidth]{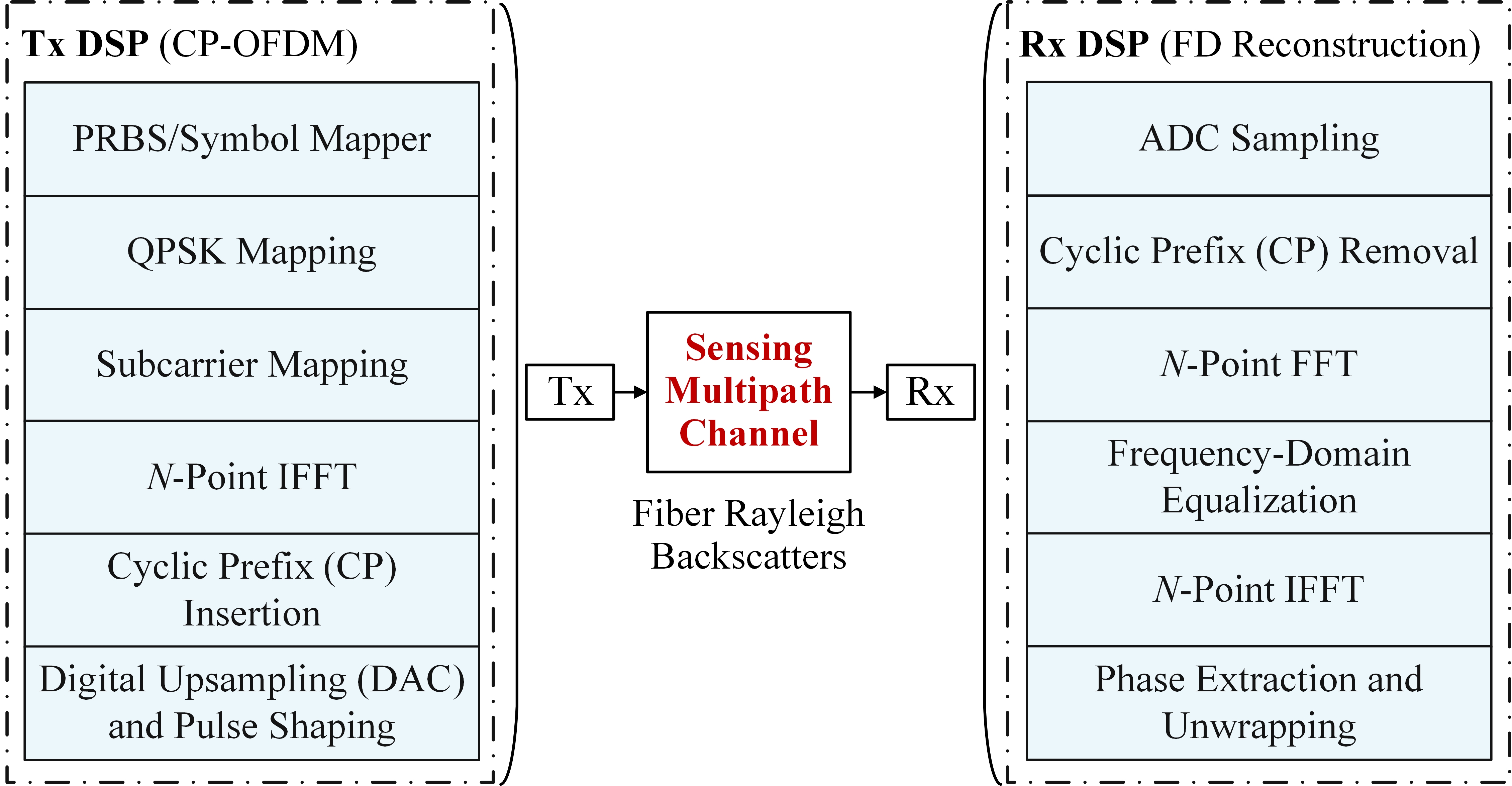}
    \caption{DSP procedure of the proposed CP-OFDM probing for matched-filter-free DAS with frequency-domain (FD) channel reconstruction.}
    \label{fig:cpofdm_dsp}
\end{figure}

The useful OFDM probing block is generated by an IDFT, following the standard baseband OFDM construction used in wireless and optical communications~\cite{GoldsmithBook,ArmstrongJLT2009}:
\begin{equation}
    s[n]
    =
    \frac{1}{N}
    \sum_{k=0}^{N-1}
    S[k]e^{j2\pi kn/N},
    \label{eq:ofdm_symbol}
\end{equation}
where $n=0,1,\ldots,N-1$.
A cyclic prefix of length \(L_{\rm CP}\) is then appended to form
\begin{equation}
    s_{\rm CP}[n]
    =
    \begin{cases}
    s[n+N], & n=-L_{\rm CP},\ldots,-1,\\
    s[n],   & n=0,\ldots,N-1.
    \end{cases}
    \label{eq:cp_signal}
\end{equation}

The transmitted probing sequence in the \(p\)-th probing period is therefore $x_p[n]=s_{\rm CP}[n]$, 
where the same OFDM probing block is used for channel reconstruction over different slow-time indices. 
The dynamic acoustic perturbation is represented by the slow-time differential phase evolution of the distributed channel taps \(h_{p,m}\).

Substituting the CP-OFDM probing sequence into the received signal model in~\eqref{eq:discrete_backscatter} gives
\begin{equation}
    y_p[n]
    =
    \sum_{m=0}^{M-1}
    h_{p,m}s_{\rm CP}[n-m]
    +
    w_p[n].
    \label{eq:cpofdm_rx_linear}
\end{equation}
The CP length satisfies $L_{\rm CP}\ge M-1$, so that the cyclic prefix fully covers the memory of the distributed Rayleigh-backscatter channel. {The sufficient CP condition, i.e., $L_{\rm CP}\ge M-1$, is not a fundamental requirement for identifying a finite-memory sensing multipath channel. In~\cite{HuangDMTISAC2026}, the authors have demonstrated the feasibility of CP-free channel reconstruction using regularized least squares (LS) estimation.}

After discarding the CP, the receiver selects the \(N\)-sample useful OFDM observation window and obtains $u_p[n]=y_p[n]$ for $n=0,1,\ldots,N-1$. The time origin is set at the beginning of the CP-removed window. Because $L_{\rm CP}\ge M-1$, the linear convolution in \eqref{eq:cpofdm_rx_linear} becomes a circular convolution over the useful OFDM block~\cite{ShiehOFDMBook,GoldsmithBook}
\begin{equation}
    u_p[n]
    =
    \sum_{m=0}^{M-1}
    h_{p,m}s[(n-m)_N]
    +
    w_p[n],
    \label{eq:circular_conv}
\end{equation}
where $n=0,\ldots,N-1$ and \((\cdot)_N\) denotes modulo-\(N\) indexing.

Taking the \(N\)-point discrete Fourier transform (DFT) of \eqref{eq:circular_conv} yields~\cite{ShiehOFDMBook,GoldsmithBook}
\begin{equation}
    U_p[k]
    =
    S[k]H_p[k]+W_p[k],
    \label{eq:freq_domain_model}
\end{equation}
where $k=0,1,\ldots,N-1$ and 
\begin{equation}
    H_p[k]
    =
    \sum_{m=0}^{M-1}
    h_{p,m}e^{-j2\pi km/N}
    \label{eq:channel_freq_response}
\end{equation}
represents the frequency response of the distributed Rayleigh-backscatter channel, and \(W_p[k]\) is the DFT of \(w_p[n]\).

Because \(S[k]\) is known and nonzero, the zero-forcing (ZF) channel estimation is~\cite{ShiehOFDMBook,GoldsmithBook}
\begin{equation}
    \hat{H}_p[k]
    =
    \frac{U_p[k]}{S[k]}.
    \label{eq:zf_channel_est}
\end{equation}
For noisy measurements or nonflat front-end responses, a regularized ZF or minimum mean-square error (MMSE) estimator can be used:
\begin{equation}
    \hat{H}_p[k]
    =
    \frac{S^H[k]}{|S[k]|^2+\lambda}
    U_p[k],
    \label{eq:mmse_channel_est}
\end{equation}
where \(\lambda\ge0\) is a regularization parameter. In the following analysis, \eqref{eq:zf_channel_est} is used to reveal the spatial-ISI-free property of our proposed CP-OFDM probing in closed form.

The discrete range-domain channel is then reconstructed by the IDFT, i.e., 
\begin{equation}
    \hat{h}^{\rm OFDM}_{p,m}
    =
    \frac{1}{N}
    \sum_{k=0}^{N-1}
    \hat{H}_p[k]e^{j2\pi km/N},
    \label{eq:ifft_channel_reconstruction}
\end{equation}
where $m=0,1,\ldots,N-1$. Under \(N\ge M\), the first \(M\) samples correspond to the discretized fiber range bins. Unlike the matched-filter-based channel estimation in \eqref{eq:mf_channel_est}, \eqref{eq:ifft_channel_reconstruction} reconstructs the channel through frequency-domain inversion rather than pulse compression. Therefore, it does not contain the deterministic sidelobe-leakage term in~\eqref{eq:mf_channel_est}, i.e., the waveform-induced spatial ISI.

Following range-domain channel reconstruction, the CP-OFDM and the LFM receivers use the same frequency-diversity combining and gauge-differential phase recovery. Let \(\hat h_{p,m}^{(q)}\), \(q=0,\ldots,Q\) denote the complex frequency-view responses obtained from either reconstructed channel, with \(q=0\) denoting the full-band view. The frequency-diversity combining channel is
\begin{equation}
    \bar h_{p,m}
    =\sum_{q=0}^{Q}c_m^{(q)}\hat h_{p,m}^{(q)},
    \label{eq:fixed_channel_combining}
\end{equation}
where the complex coefficients \(c_m^{(q)}\) are estimated from the dominant component of the noise-whitened and uncentered frequency-view covariance, and remain fixed throughout each continuous acquisition. Without diversity, \(Q=0\), \(c_m^{(0)}=1\), and then \(\bar h_{p,m}=\hat h_{p,m}\). For a gauge separation of \(G\) samples, the complex gauge product is
\begin{equation}
    \hat g_{p,\ell}
    =\bar h_{p,\ell+G}\left(\bar h_{p,\ell}\right)^H.
    \label{eq:common_gauge_product}
\end{equation}
Let \(\mathcal I_{\ell,r}\) denote the \(r\)-th continuous reliable slow-time interval at gauge position \(\ell\), and let \(p_{\ell,r}=\min\mathcal I_{\ell,r}\). The differential phase is
\begin{equation}
    \Delta\hat\phi_{p,\ell}
    =\operatorname{unwrap}_{p\in\mathcal I_{\ell,r}}
    \angle\!\left\{
    \hat g_{p,\ell}\left(\hat g_{p_{\ell,r},\ell}\right)^H
    \right\}.
    \label{eq:common_phase_recovery}
\end{equation}
Each reliable interval has an independent phase reference; missing observations are not interpolated or connected across gaps. 
The CP-OFDM and conventional LFM receivers obtain their channel responses \(\hat h_{p,m}^{\rm OFDM}\) and \(\hat h_{p,m}^{\rm MF}\) through CP-aided frequency-domain channel reconstruction and matched-filter-based pulse compression, respectively, whereas their subsequent phase processing is the same.

\subsection{Spatial-ISI-Free Channel Reconstruction}
\label{CPOFDMTheorem}
In the proposed CP-OFDM probing, the discrete spatial-domain channel is reconstructed using~\eqref{eq:ifft_channel_reconstruction}. 
The following theorem formalizes the spatial-ISI-free property of the CP-OFDM channel reconstruction.

\begin{theorem} 
\label{thm:spatial_isi_free}
Consider the distributed Rayleigh-backscatter model in \eqref{eq:discrete_backscatter}. Assume that: 
\begin{itemize}
    \item the distributed channel \(h_{p,m}\), \(m=0, 1, \ldots,M-1\), is quasi-static within one OFDM probing block;
    \item the useful OFDM length satisfies \(N\ge M\) and the CP length satisfies \(L_{\rm CP}\ge M-1\);
    \item all probing weights are known and nonzero, i.e., \(S[k]\neq0\) for \(k=0, 1, \ldots,N-1\).
\end{itemize}
Then the ZF channel reconstruction in \eqref{eq:zf_channel_est} and \eqref{eq:ifft_channel_reconstruction} yields
\begin{equation}
    \hat{h}^{\rm OFDM}_{p,m}
    =
    h_{p,m}
    +
    \tilde{w}_{p,m},
    \label{eq:ofdm_reconstruction_result}
\end{equation}
where $m=0,1,\ldots,M-1$ and 
\begin{equation}
    \tilde{w}_{p,m} = \frac{1}{N}
    \sum_{k=0}^{N-1}
    \frac{W_p[k]}{S[k]}
    e^{j2\pi km/N}
    \label{eq:ofdm_noise_term}
\end{equation}
is the reconstructed noise term. In the absence of noise, \(\hat{h}^{\rm OFDM}_{p,m}=h_{p,m}\) exactly. The reconstruction contains no deterministic sidelobe leakage, i.e.,  
\(\sum_{q\neq m} h_{p,q}c_s[m-q]\), and is therefore spatial-ISI-free.
\end{theorem}

\begin{IEEEproof}
Under \(L_{\rm CP}\ge M-1\), removing the CP transforms the linear convolution between the transmitted CP-OFDM sequence and the \(M\)-tap distributed backscatter channel into the circular convolution in \eqref{eq:circular_conv}. Taking the \(N\)-point DFT gives
\begin{equation}
    U_p[k]=S[k]H_p[k]+W_p[k].
    \label{eq:proof_freq_model}
\end{equation}
Since \(S[k]\neq0\), the ZF estimate satisfies
\begin{equation}
    \hat{H}_p[k]
    =
    \frac{U_p[k]}{S[k]}
    =
    H_p[k]
    +
    \frac{W_p[k]}{S[k]}.
    \label{eq:proof_zf}
\end{equation}
Substituting~\eqref{eq:proof_zf} into~\eqref{eq:ifft_channel_reconstruction} gives
\begin{equation}
    \hat{h}^{\rm OFDM}_{p,m}
   =
    \frac{1}{N}\!\!
    \sum_{k=0}^{N-1}\!\!
    H_p[k]e^{j2\pi km/N}
    \!+\!
    \frac{1}{N}\!\!
    \sum_{k=0}^{N-1}\!\!
    \frac{W_p[k]}{S[k]}e^{j2\pi km/N}.
    \label{eq:proof_ifft}
\end{equation}
Substituting \eqref{eq:channel_freq_response} into the first term of \eqref{eq:proof_ifft}, we have
\begin{align}
    \frac{1}{N}\!\!
    \sum_{k=0}^{N-1}\!\!
    H_p[k]e^{j2\pi km/N}
    &=\! 
    \frac{1}{N}\!\!
    \sum_{k=0}^{N-1}\! 
    \sum_{q=0}^{M-1}\!\!
    h_{p,q}
    e^{-j2\pi kq/N}
    e^{j2\pi km/N}  \label{eq:proof_delta1} \\
    &=\! 
    \sum_{q=0}^{M-1}\! 
    h_{p,q}\!\!
    \left[
    \frac{1}{N}\!\!
    \sum_{k=0}^{N-1}
    e^{j2\pi k(m-q)/N}
    \right].
    \label{eq:proof_delta}
\end{align}
The term in brackets is the periodic Kronecker delta,
\begin{equation}
    \frac{1}{N}
    \sum_{k=0}^{N-1}
    e^{j2\pi k(m-q)/N}
    =
    \delta_N[m-q].
    \label{eq:periodic_delta}
\end{equation}
Because \(N\ge M\) and \(m,q\in\{0,\ldots,M-1\}\), the periodic delta reduces to the ordinary Kronecker delta over the physical range-bin support. Hence,
\begin{equation}
    \sum_{q=0}^{M-1}
    h_{p,q}\delta_N[m-q]
    =
    h_{p,m},
    \label{eq:proof_channel_exact}
\end{equation}
which gives \eqref{eq:ofdm_reconstruction_result}. Unlike the matched-filter-based channel estimation in \eqref{eq:mf_channel_est}, no summation over \(q\neq m\) remains. This proves that the CP-OFDM reconstruction is free from deterministic waveform-induced spatial ISI.
\end{IEEEproof}

\begin{figure}[!t]
    \centering
    {
        \includegraphics[width=0.9\linewidth]{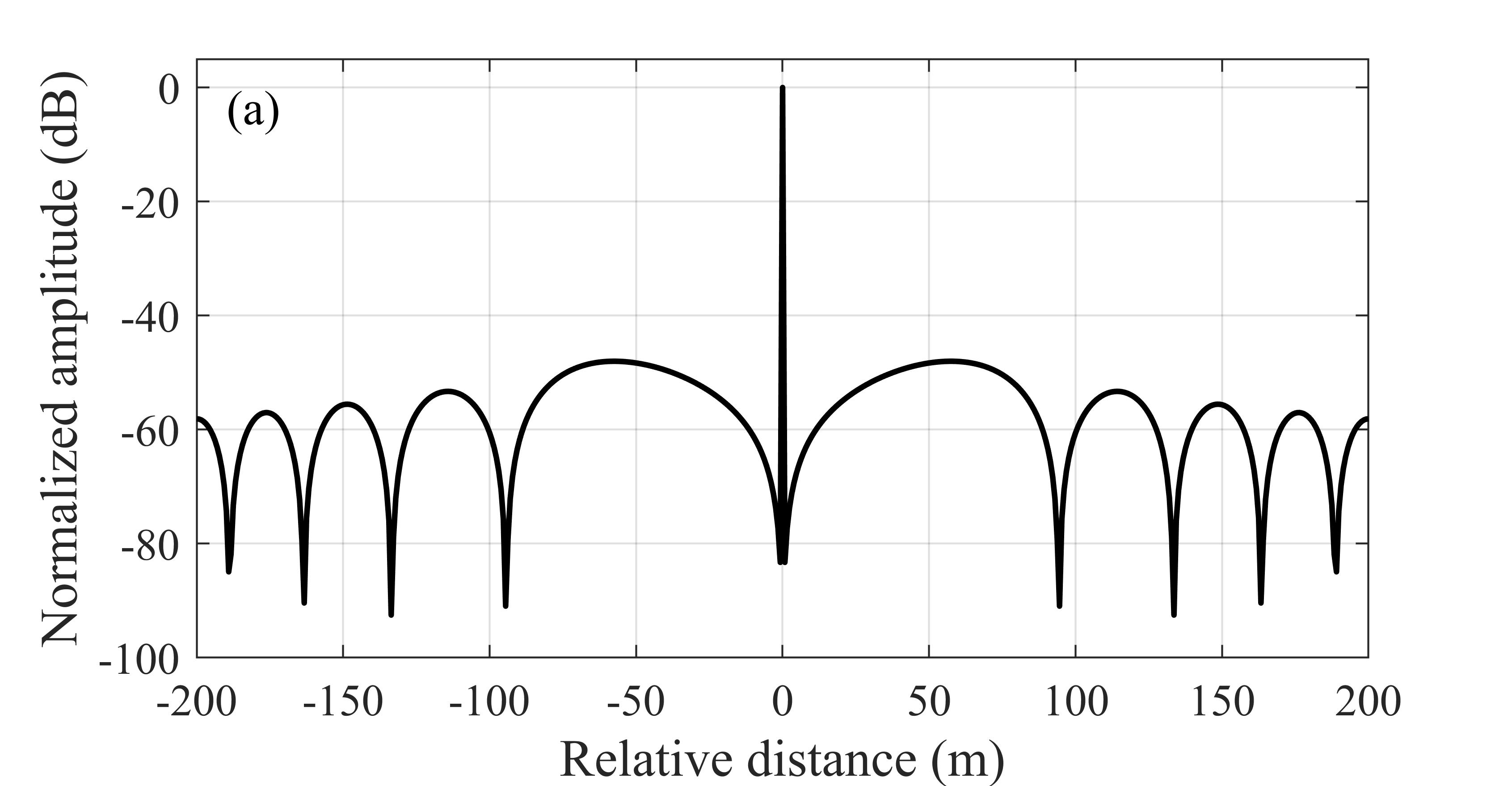}
        \label{fig:psf_lfm}
    }\\ 
    {
        \includegraphics[width=0.9\linewidth]{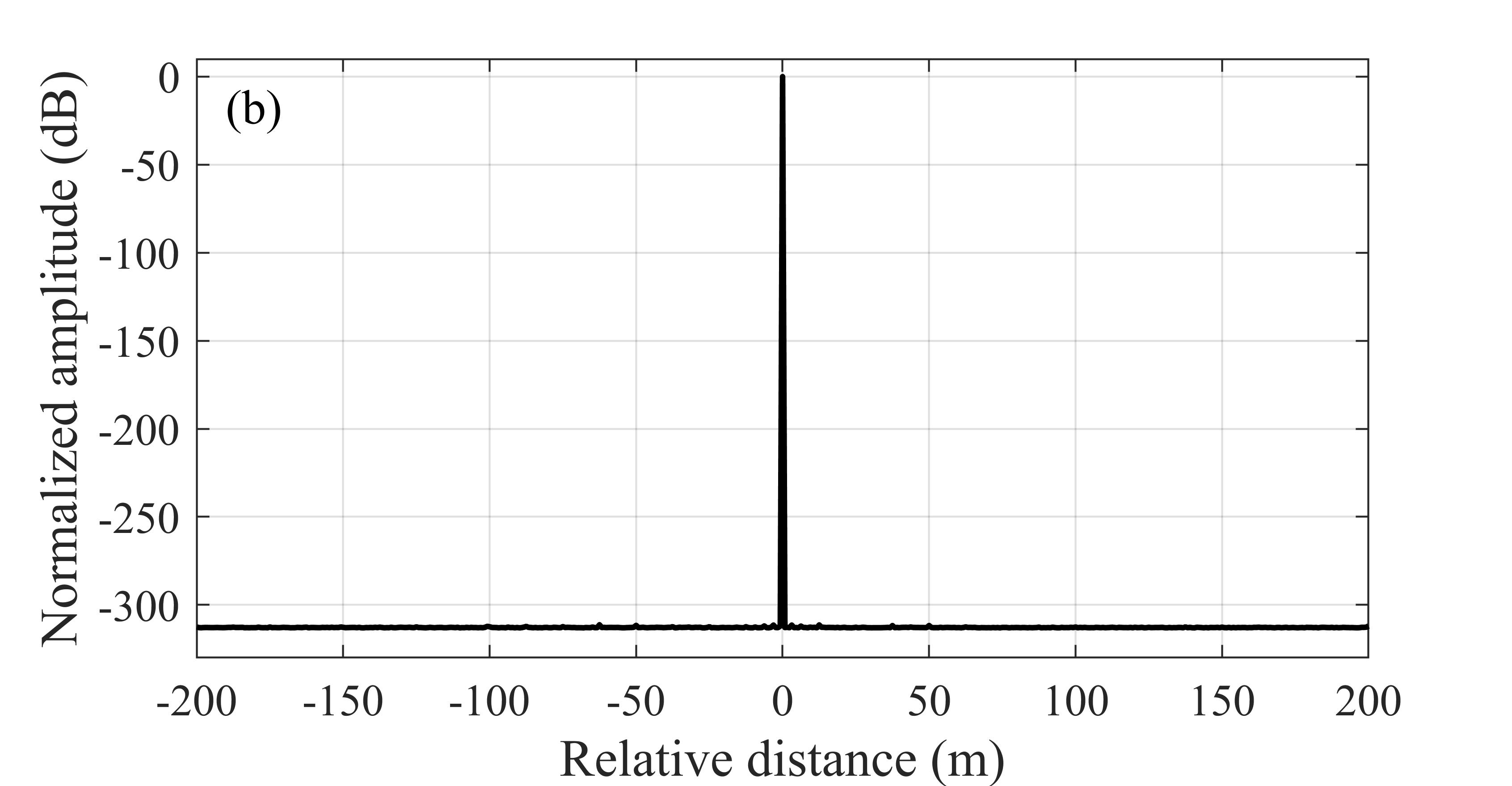}
        \label{fig:psf_cpofdm}
    }
    \caption{Range-bin reconstruction response for a grid-aligned point scatterer using (a) the conventional matched-filter-based LFM time-domain pulse-compression receiver and (b) the proposed CP-OFDM-aided frequency-domain channel-reconstruction receiver under the same occupied bandwidth and in the absence of noise.}
    \label{fig:psf_comparison}
\end{figure}

Fig.~\ref{fig:psf_comparison} compares the two reconstruction responses. For LFM probing, the normalized autocorrelation response in~\eqref{eq:mf_channel_est} gives a finite pulse-compression response with nonzero sidelobes even for an isolated range bin, which cause deterministic leakage in dense-event DAS. The CP-OFDM receiver instead reconstructs the backscatter channel through CP removal, frequency-domain equalization, and IDFT.

According to Theorem~\ref{thm:spatial_isi_free}, the on-grid response of the CP-OFDM channel reconstruction is a Kronecker delta. The residual floor in Fig.~\ref{fig:psf_comparison}~(b) arises from numerical precision rather than from the pulse-compression sidelobes in Fig.~\ref{fig:psf_comparison}~(a). Thus, CP-aided frequency-domain channel reconstruction eliminates deterministic waveform-induced leakage among range bins. It does not remove receiver noise, Rayleigh or polarization fading, gauge overlap, or physical overlap between perturbations, which still affect \eqref{eq:common_gauge_product} and \eqref{eq:common_phase_recovery}. 

For the full-band, single-view case, substituting \eqref{eq:ofdm_reconstruction_result} into the common gauge operation in~\eqref{eq:common_gauge_product} gives
\begin{equation}
    \hat{g}^{\rm OFDM}_{p,\ell}
    =
    \left(h_{p,\ell+G}+\tilde{w}_{p,\ell+G}\right)
    \left(h_{p,\ell}+\tilde{w}_{p,\ell}\right)^H,
    \label{eq:ofdm_gauge_noise_expansion}
\end{equation}
which contains noise-dependent perturbations, but not the autocorrelation-induced leakage of the form shown in \eqref{eq:mf_error}. In this single-view case, the recovered slow-time differential phase contains no deterministic pulse-compression sidelobe leakage from other fiber range bins. Theorem~\ref{thm:spatial_isi_free} concerns the channel reconstruction before diversity combining; subsequent frequency weighting changes the effective spatial response.



\subsection{Spatial Resolution Analysis}\label{SpatialResolu}
The matched-filter-based LFM receiver obtains range resolution through pulse compression, whereas the CP-OFDM receiver reconstructs the range bins through frequency-domain equalization followed by an IDFT.
For the LFM receiver, the matched-filter-based channel estimate is given by~\eqref{eq:mf_channel_est}.
Hence, the LFM receiver does not directly recover \(h_{p,\ell}\); instead, it recovers the channel blurred by the autocorrelation kernel \(c_s[\cdot]\). 
The main lobe of \(c_s[\cdot]\) determines the compressed-pulse
spatial resolving interval, while the sidelobes of \(c_s[\cdot]\) produce
deterministic leakage from neighboring range bins, which is the same pulse-compression sidelobe mechanism treated in radar and optical pulse-compression DAS~\cite{RichardsRadar,MompoOL2018}.

To illustrate this dependence, consider the continuous-delay autocorrelation kernel of an LFM probe with a nearly rectangular spectral magnitude across the occupied bandwidth \(B_{\rm occ}\).
After normalization to unit peak, the compression kernel is
\begin{equation}
    c_s(\tau)
    =
    \frac{1}{B_{\rm occ}}
    \int_{-B_{\rm occ}/2}^{B_{\rm occ}/2}
    e^{j2\pi f\tau}\,df
    =
    \operatorname{sinc}(B_{\rm occ}\tau),
    \label{eq:lfm_rect_kernel}
\end{equation}
where \(\operatorname{sinc}(x)=\sin(\pi x)/(\pi x)\). The first nulls occur at $\tau = \pm \frac{1}{B_{\rm occ}}$.
For a backscatter DAS system, a delay interval \(\Delta \tau\) corresponds to a
spatial interval \(v_g\Delta\tau/2\). Therefore, the one-sided first-null spatial interval, measured from the peak of the LFM compressed response, is
\begin{equation}
    \Delta z_{\rm LFM}
    =
    \frac{v_g}{2B_{\rm occ}}.
    \label{eq:lfm_first_null_spatial}
\end{equation}
Note that different time-domain or frequency-domain weightings change the exact main-lobe width and sidelobe level of \(c_s(\tau)\). In particular, windowing can reduce the sidelobes of the pulse-compression response, but it broadens the main lobe and therefore increases the effective resolving interval, as explicitly analyzed for optical pulse-compression reflectometry in~\cite{MompoOL2018}.

The proposed CP-OFDM probing uses frequency-domain channel reconstruction. After CP removal, the useful OFDM block satisfies the circular-convolution model in~\eqref{eq:circular_conv}, whose DFT-domain representation is given by~\eqref{eq:freq_domain_model}. With known nonzero subcarrier weights, the channel frequency samples are obtained by the one-tap operation in~\eqref{eq:zf_channel_est} or~\eqref{eq:mmse_channel_est}.
The spatial-domain channel is then reconstructed by the IDFT in \eqref{eq:ifft_channel_reconstruction}.

For an aperture of  \(N_{\rm act}\le N\) contiguous active subcarriers with spacing \(\Delta f=1/T_u\), where \(T_u\) is the duration of a useful OFDM symbol, the occupied bandwidth is
\begin{equation}
    B_{\rm occ}
    =N_{\rm act}\Delta f.
    \label{eq:ofdm_frequency_aperture_res}
\end{equation}
The receiver retains the \(N\)-point DFT grid of the useful symbol, with \(F_s=N\Delta f\). The estimated channel samples are placed at their original frequency indices, and the inactive frequency samples are set to zero. The \(N\)-point IDFT then provides the delay grid
\begin{equation}
    \tau_m
    =\frac{m}{N\Delta f}
    =\frac{m}{F_s},
    \label{eq:ofdm_delay_grid}
\end{equation}
where \(m=0,1,\ldots,N-1\). The adjacent reconstruction-sample spacing is
\begin{equation}
    \Delta\tau_{\rm OFDM}
    =\tau_{m+1}-\tau_m
    =\frac{1}{N\Delta f}
    =\frac{1}{F_s}.
    \label{eq:ofdm_delay_bin_res}
\end{equation}
The corresponding spatial sampling interval is \(v_g/(2F_s)\). In contrast, the bandwidth-limited spatial resolving interval for the rectangular active-frequency aperture is
\begin{equation}
    \Delta z_{\rm OFDM}
    =\frac{v_g}{2N_{\rm act}\Delta f}
    =\frac{v_g}{2B_{\rm occ}}.
    \label{eq:ofdm_spatial_bin_res}
\end{equation}
Therefore, \(\Delta z_{\rm OFDM}\) denotes the resolving interval, whereas \(v_g/(2F_s)\) denotes the reconstruction sampling interval. These intervals coincide in the critically sampled full-band case \(N_{\rm act}=N\). The subcarrier spacing determines the periodic delay span, \(\tau_{\rm per}=1/\Delta f=T_u\). A finer reconstruction grid does not increase the occupied bandwidth or improve the bandwidth-limited resolution.

To isolate the finite-aperture response, consider noiseless ZF reconstruction for the contiguous active set \(k=0,1,\ldots,N_{\rm act}-1\), without subsequent frequency-diversity weighting. A single scatterer with complex coefficient \(\alpha\) and continuous delay \(\tau_0\) has the following frequency response 
\begin{equation}
    H[k]
    =\alpha e^{-j2\pi k\Delta f\tau_0}.
    \label{eq:single_scatterer_freq_res}
\end{equation}
Its reconstructed response on the \(N\)-point delay grid is
\begin{equation}
    \hat h[m]
    =\frac{\alpha}{N}
    \sum_{k=0}^{N_{\rm act}-1}
    e^{j2\pi k\left(\frac{m}{N}-\Delta f\tau_0\right)}.
    \label{eq:single_scatterer_idft_res}
\end{equation}
Let
\begin{equation}
    \epsilon_m
    =\frac{m}{N}-\Delta f\tau_0.
    \label{eq:ofdm_epsilon_res}
\end{equation}
Then
\begin{equation}
    \hat h[m]
    =\frac{\alpha}{N}
    e^{j\pi(N_{\rm act}-1)\epsilon_m}
    \frac{\sin(\pi N_{\rm act}\epsilon_m)}{\sin(\pi\epsilon_m)},
    \label{eq:ofdm_dirichlet_res}
\end{equation}
with the ratio evaluated by continuity when its denominator vanishes. The continuous-delay envelope of this Dirichlet response has its first null at an offset \(1/(N_{\rm act}\Delta f)=1/B_{\rm occ}\) from the scatterer, yielding~\eqref{eq:ofdm_spatial_bin_res}. The response is sampled at the interval given by~\eqref{eq:ofdm_delay_bin_res}. Shifting the contiguous active set introduces a phase ramp without altering this magnitude response. For \(N_{\rm act}<N\), even a scatterer aligned to the \(N\)-point delay grid generally has nonzero neighboring samples. The Kronecker delta result of Theorem~\ref{thm:spatial_isi_free} applies to an on-grid scatterer in the full-band ZF case, when \(N_{\rm act}=N\), subject to the stated channel and CP assumptions.

Note that our CP-OFDM reconstructs the channel through frequency-domain inversion, whereas LFM matched filtering retains the probing-waveform autocorrelation kernel. Their nominal bandwidth-limited resolving intervals are the same under the same occupied bandwidth. However, their exact spatial responses also depend on the spectral masks and receiver weightings.

\section{Simulations and Discussion}\label{Simu}
To validate the spatial-ISI-free channel reconstruction of the proposed CP-OFDM probing, a $5.2$-km $\phi$-OTDR DAS link was simulated using the discrete Rayleigh-backscatter model in \eqref{eq:discrete_backscatter}. 
The effective refractive index was set to $n_{\rm eff}=1.5$, corresponding to $v_g=2.0\times10^8$~m/s. The amplitude attenuation coefficient was $\alpha=2.3\times10^{-5}$~m$^{-1}$, and the complex receiver noise was modeled as additive white Gaussian noise. 
Both the matched-filter-based LFM and CP-OFDM branches used an effective complex-baseband bandwidth of $128$~MHz and a sampling rate of $128$~MSa/s. The resulting range-bin spacing was $\Delta z=v_g/(2F_s)=0.781$~m. Both receivers used the single-view form of the common phase-recovery procedure in \eqref{eq:common_gauge_product} and \eqref{eq:common_phase_recovery}, with $G=3$ samples, i.e., $2.344$~m.
 
The conventional branch used an energy-normalized LFM pulse, the matched-filter-based channel estimation in \eqref{eq:mf_channel_est}, and the gauge-differential phase recovery in \eqref{eq:common_phase_recovery}. The proposed branch used CP-OFDM probing with $N=8192$ active QPSK-weighted subcarriers and $L_{\rm CP}=6655$ CP samples. The simulated fiber contained $M=6656$ range bins; hence, the conditions $N\ge M$ and $L_{\rm CP}\ge M-1$ in Theorem~\ref{thm:spatial_isi_free} were satisfied. After CP removal, the proposed CP-OFDM-aided frequency-domain channel-reconstruction receiver obtained the one-tap estimate in \eqref{eq:zf_channel_est} and reconstructed the channel using \eqref{eq:ifft_channel_reconstruction}.
Consequently, the proposed branch contained no deterministic spatial-ISI of the form shown in \eqref{eq:mf_channel_est}. The complete CP-OFDM block duration was $115.99$~$\mu$s, and the fiber round-trip time was $52.00$~$\mu$s. The probing repetition period was $201.59$~$\mu$s, yielding a slow-time sampling rate of $4.96$~kHz. A total of $150$ probing periods were simulated, corresponding to a monitoring interval of approximately $30$~ms.

To stress the spatial-ISI term in \eqref{eq:mf_gauge_expansion}, ten simultaneous events were arranged into groups of $3$, $4$, and $3$ near $1.05$, $3.00$, and $4.50$~km. Their intragroup spacings were $5.47$, $5.83$, and $5.31$~m, respectively. Four strong and six weak events had amplitudes of $0.12$--$1.00$~rad and frequencies of $279.8$--$410.4$~Hz around group centers of $280$, $340$, and $410$~Hz. This setting exposed matched-filter sidelobe leakage without dominant gauge overlap.

\begin{figure}[!t]
    \centering
    {
        \includegraphics[width=0.94\linewidth]{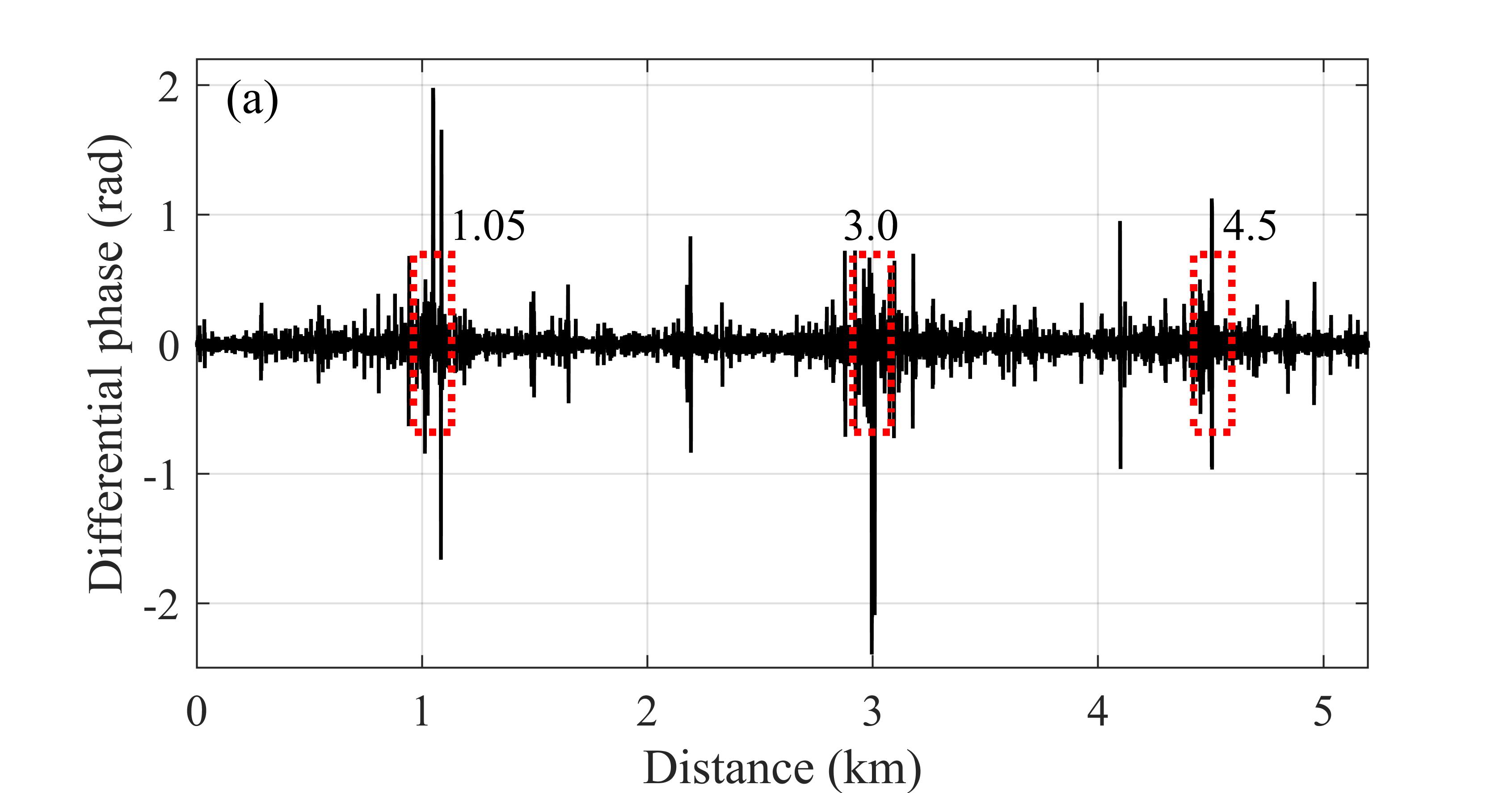}
        \label{fig:Gphi_lfm}
    }\\ 
    {
        \includegraphics[width=0.94\linewidth]{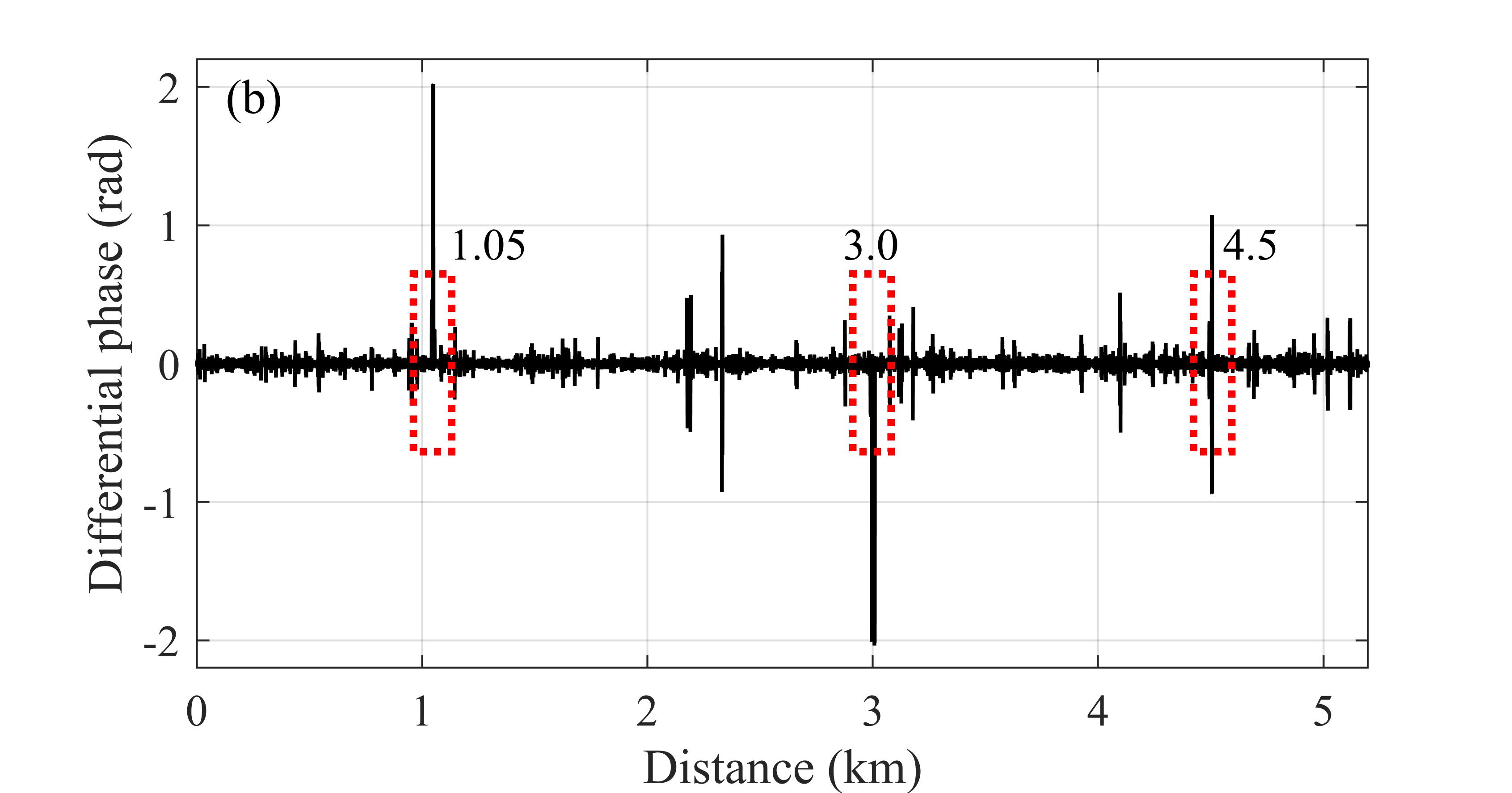}
        \label{fig:Gphi_cpofdm}
    }
    \caption{Gauge-differential phase traces over the whole sensing fiber and all $150$ probing periods using (a) the LFM matched-filter receiver and (b) the proposed CP-OFDM-aided frequency-domain channel-reconstruction receiver.}
    \label{fig:Gphi_comparison}
\end{figure}

Fig.~\ref{fig:Gphi_comparison} shows that the matched-filter-based LFM receiver spreads the phase responses spatially and produces visible off-event fluctuations, particularly near the ten events. This agrees with \eqref{eq:mf_channel_est}, where \(c_s[\ell-m]\) couples neighboring bins before the nonlinear phase extraction in \eqref{eq:common_phase_recovery}. The CP-OFDM receiver instead confines the responses to the preset event groups and suppresses the off-event floor, consistent with Theorem~\ref{thm:spatial_isi_free}. Its residual fluctuations are noise-dependent rather than deterministic pulse-compression sidelobes.

\begin{figure}[!t]
    \centering
    \includegraphics[width=0.42\textwidth]{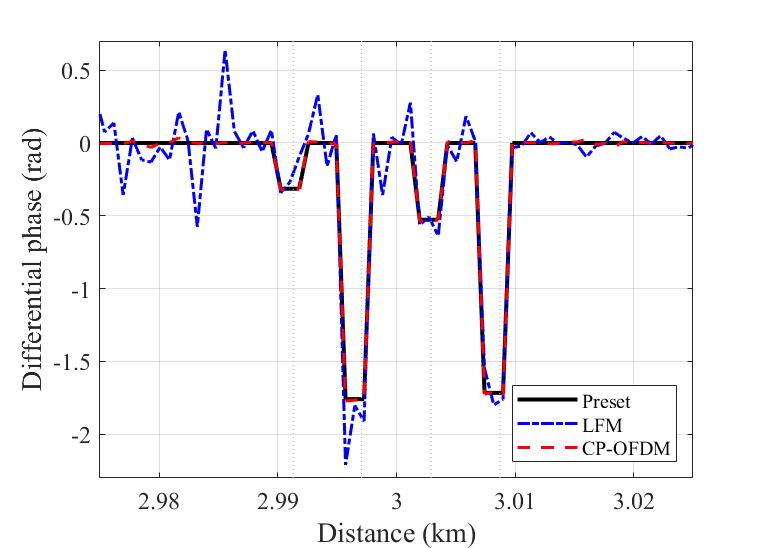}
    \caption{Gauge-differential phase snapshot at one probing period for the four closely spaced events around $3$~km, comparing the preset phase, the matched-filter-based LFM receiver, and the proposed CP-OFDM-aided frequency-domain channel-reconstruction receiver.}
    \label{fig:oneperiod}
\end{figure}

\begin{figure*}[!t]
    \centering
    \includegraphics[width=0.96\textwidth]{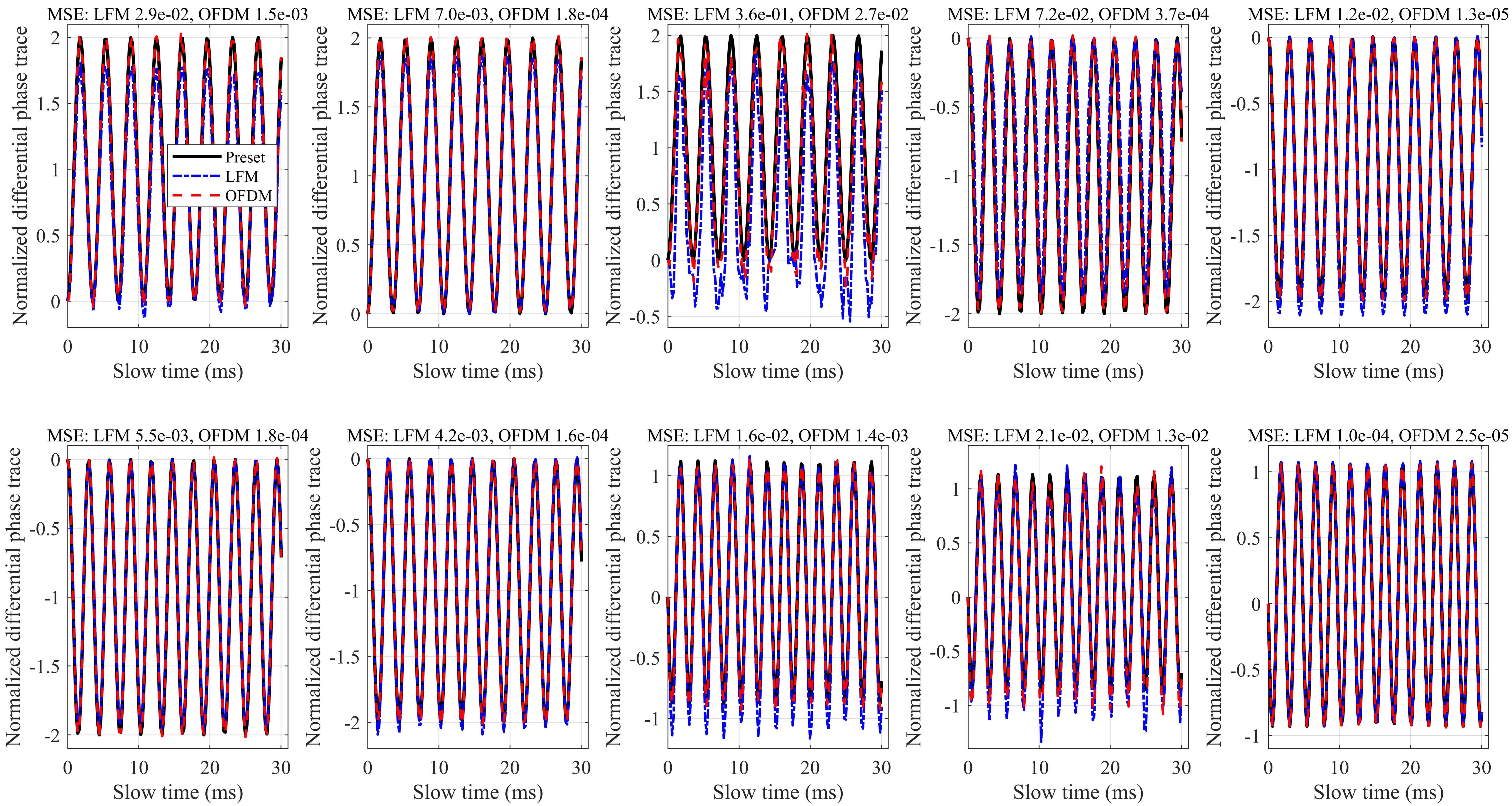}
    \caption{Normalized slow-time phase traces at the ten event positions, comparing the preset phase trace, the matched-filter-based LFM receiver, and the proposed CP-OFDM-aided frequency-domain channel-reconstruction receiver, where the preset phase trace for each event is normalized to a peak-to-peak value of $2$, and the LFM and CP-OFDM traces use the same scale factor.}
    \label{fig:recovered_phase}
\end{figure*}

Fig.~\ref{fig:oneperiod} examines four events separated by $5.83$~m near $3$~km at the 83rd probing period. The LFM result deviates from the preset trace because \(h_{p,m}c_s[\ell-m]\) leaks strong-event responses into neighboring weak-event bins. The CP-OFDM result more closely follows the preset trace because \eqref{eq:ofdm_reconstruction_result} contains no summation over \(q\ne m\), confirming reduced strong-to-weak leakage (i.e., the waveform-induced spatial ISI) at the same bandwidth-limited range-bin spacing.

Fig.~\ref{fig:recovered_phase} compares the slow-time phase traces obtained by applying \eqref{eq:common_phase_recovery} to the two reconstructed channels. Each recovered trace is plotted using the scale factor determined from the corresponding preset trace, preserving amplitude bias and waveform distortion.
For the $n$-th event, the normalized mean-square error (MSE) annotated in the corresponding subplot is calculated from the same preset-normalized traces:
\begin{equation*}
{\rm MSE}_{n}^{(\cdot)}
=\frac{1}{P}\sum_{p=0}^{P-1}
\left(
\Delta\hat{\phi}_{p,\ell}^{(\cdot)}
-\Delta\phi_{p,\ell}^{\rm pre}
\right)^2,
\end{equation*}
where $n=1,\ldots,10$, \(P=150\), \((\cdot)\) denotes either the matched-filter-based LFM or proposed CP-OFDM receiver, and \(\Delta\phi_{p,\ell}^{\rm pre}\) is the preset phase trace after the same peak-to-peak normalization. 

The gain \(10\log_{10}({\rm MSE}_{n}^{\rm MF}/{\rm MSE}_{n}^{\rm OFDM})\) is positive when CP-OFDM has a smaller MSE. The ten gains are $12.74$, $15.92$, $11.31$, $22.91$, $29.55$, $14.87$, $14.11$, $10.67$, $1.94$, and $6.09$~dB for the ten events, respectively. Thus, CP-OFDM improves every event, particularly weak events near strong perturbations; even E9 retains a lower MSE. This agrees with \eqref{eq:mf_gauge_expansion} and \eqref{eq:ofdm_gauge_noise_expansion}. The remaining deviations arise mainly from receiver noise, Rayleigh fading, and finite gauge length rather than pulse-compression sidelobes.

\section{Experimental Verification}\label{Exp}
This section experimentally verifies the practical feasibility of the proposed CP-OFDM DAS. 
The receiver reconstructs the distributed Rayleigh-backscatter channel through frequency-domain equalization and then recovers localized vibrations from the gauge-differential phase. 
Because the frequency-domain coefficients of the CP-OFDM probe are standard communication symbols, the same waveform can also support ISAC. 
Fig.~\ref{fig:Expsch} illustrates this shared-waveform operation. 
\begin{figure}[!t]
    \centering
    \includegraphics[width=0.42\textwidth]{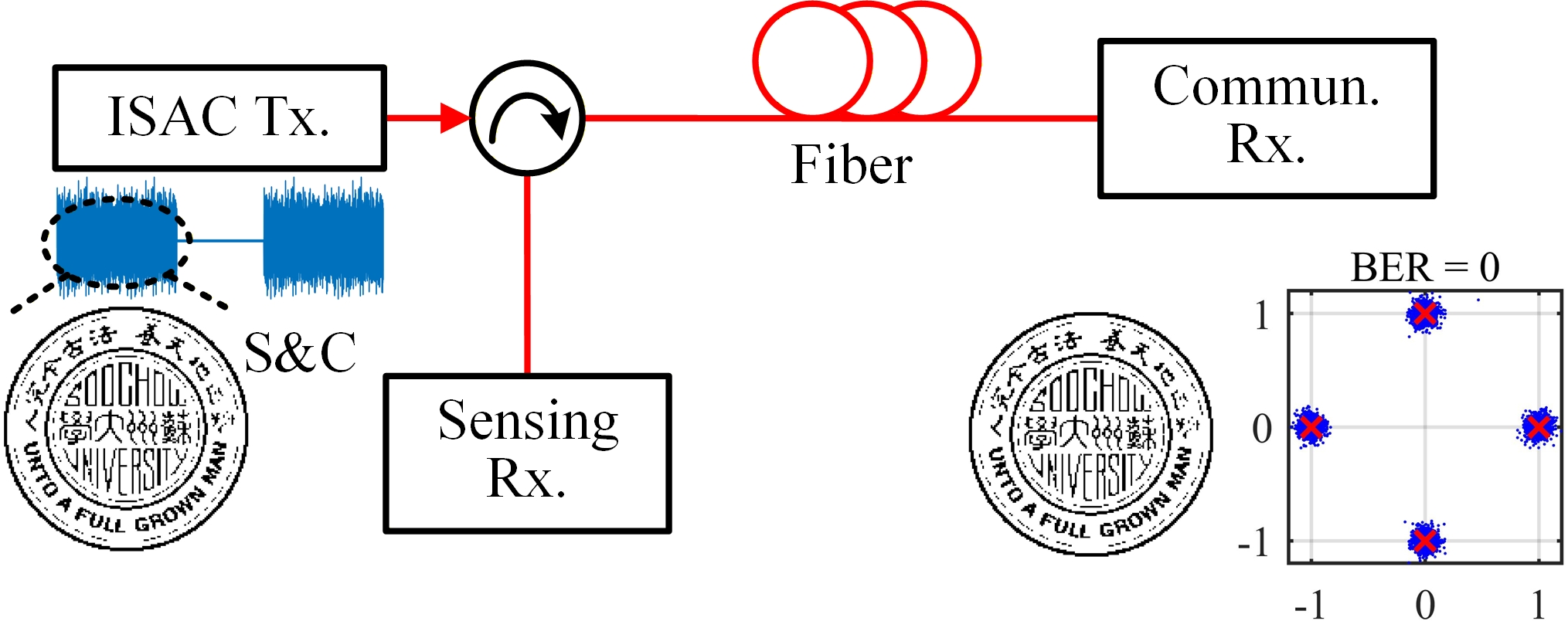}
    \caption{Illustration of shared-waveform integrated sensing and communication (ISAC) using the proposed CP-OFDM waveform for both sensing and communication (S\&C).}
    \label{fig:Expsch}
\end{figure}

Instead of using a specific sensing sequence, such as a pseudorandom binary sequence (PRBS), the input image is quantized, serialized, scrambled, and mapped onto the OFDM subcarriers as QPSK symbols. These data-bearing symbols form the known vector \(S[k]\) used to synthesize the CP-OFDM probe. At the forward communication terminal, CP removal, a fast Fourier transform (FFT), channel equalization, and QPSK decisions recover the payload.  
As shown in Fig.~\ref{fig:Expsch}, the image is recovered with zero measured bit-error rate (BER), and independent validation frames yield a median error vector magnitude (EVM) of \(-23.14\)~dB.

At the CP-OFDM sensing receiver, the same known \(S[k]\) provides the probing reference for the regularized one-tap channel reconstruction in~\eqref{eq:mmse_channel_est}, with the effective reference spectrum specified below.
Therefore, embedding communication data neither alters the CP-aided circular-convolution model nor requires an additional sensing waveform. 
As illustrated in Fig.~\ref{fig:Expsch}, the same optical CP-OFDM signal simultaneously carries communication information and probes the distributed Rayleigh-backscatter channel. 
As this experiment serves as an initial proof of concept, the following results focus on DAS localization and vibration recovery rather than a complete characterization of the communication performance.

\begin{figure}[!t]
    \centering
    \includegraphics[width=0.45\textwidth]{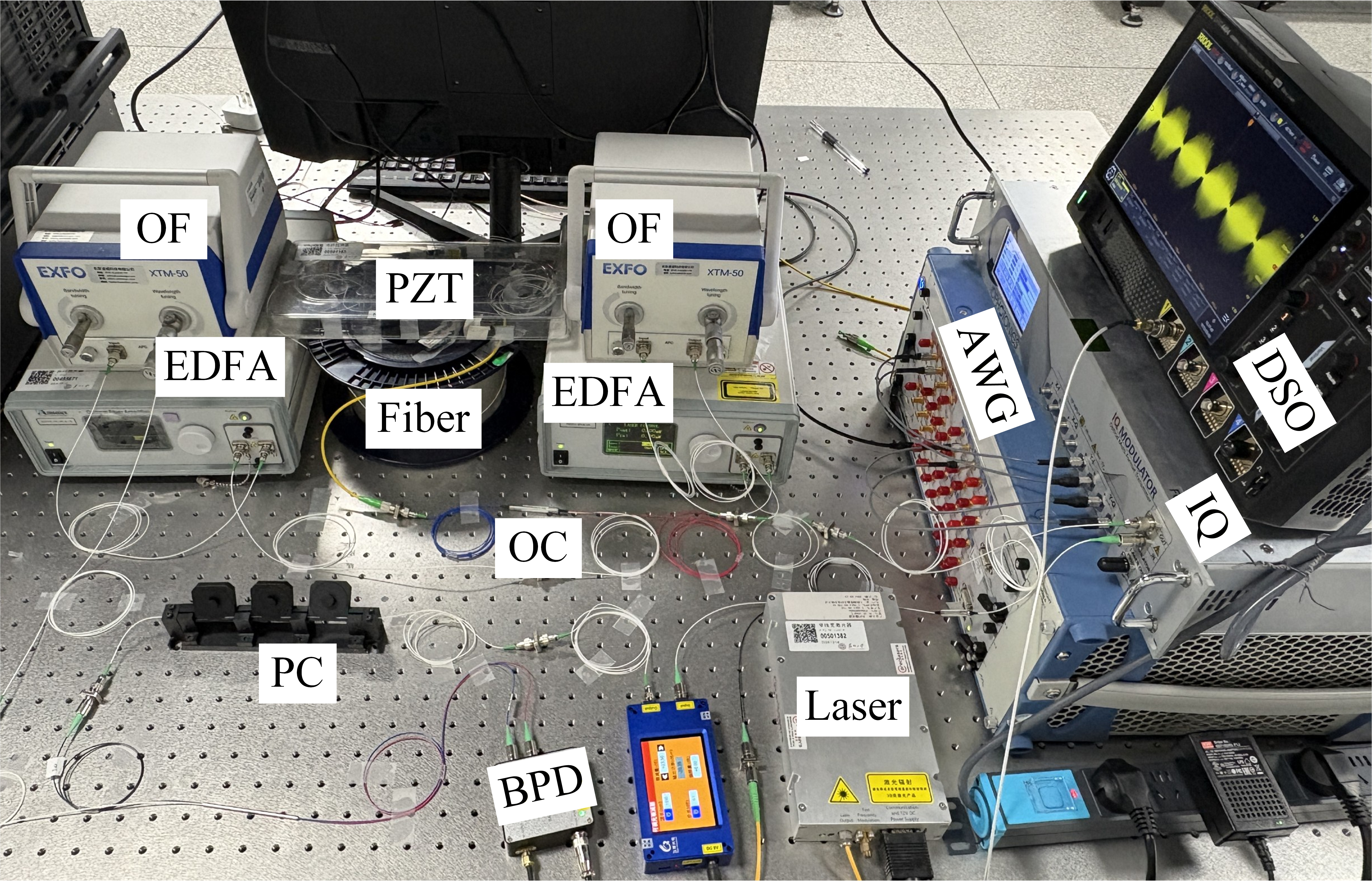}
    \caption{Experimental setup for CP-OFDM DAS and the LFM benchmark. AWG: arbitrary waveform generator; OF: optical filter; EDFA: erbium-doped fiber amplifier; PZT: piezoelectric transducer; PC: polarization controller; BPD: balanced photodetector; DSO: digital storage oscilloscope.}
    \label{fig:Exp}
\end{figure}

The CP-OFDM heterodyne coherent DAS experimental setup is shown in Fig.~\ref{fig:Exp}. The output of a narrow-linewidth Precilasers FL-SF-1550-S laser is divided into signal and LO branches. In the signal branch, a Keysight M8190A arbitrary waveform generator (AWG) generates the in-phase and quadrature electrical drives for an ID PHOTONICS OMFTV2 OMFT-C-00FA IQ modulator. The CP-OFDM waveform is generated at a complex-baseband sampling rate of \(F_s=128\)~MSa/s, digitally shifted to a \(160\)-MHz intermediate frequency, and exported by the AWG at \(512\)~MSa/s. The useful OFDM symbol contains \(N=8192\) samples and \(N_{\rm act}=7167\) active QPSK subcarriers. The subcarrier spacing is \(15.625\)~kHz, and the occupied bandwidth is \(B_{\rm occ}=111.984\)~MHz. The active index set \(\mathcal K_{\rm act}\subset\{0,1,\ldots,N-1\}\) excludes the band-edge guard subcarriers and the direct-current (DC) subcarrier. With a nominal group velocity of \(v_g=2.0\times10^8\)~m/s, the spatial sampling interval is \(v_g/(2F_s)=0.78125\)~m, and the nominal bandwidth-limited resolving interval is \(v_g/(2B_{\rm occ})\approx0.893\)~m.

The cyclic prefix contains \(7167\) samples, corresponding to \(55.992~\mu\)s. This duration covers the approximately \(52~\mu\)s round-trip delay of the \(5.2\)-km fiber and retains an additional \(4~\mu\)s timing margin, thereby satisfying the sufficient-CP condition used in Theorem~\ref{thm:spatial_isi_free}. The useful-symbol duration is \(64~\mu\)s, and the combined duration of the CP and useful symbol is \(119.992~\mu\)s. A \(1~\mu\)s cyclic suffix accommodates the pulse fade-out, giving a complete emitted pulse duration of \(120.992~\mu\)s. A repetition period of \(189.281~\mu\)s gives a slow-time sampling rate of \(5.283\)~kHz and a Nyquist frequency of \(2.642\)~kHz. An LFM benchmark using matched-filter pulse compression is implemented with the same optical setup, occupied bandwidth, complete emitted pulse duration, repetition period, and sampling rates.

After IQ modulation, the optical signal is spectrally conditioned by the first EXFO XTM-50 optical filter (OF) and amplified by an Amonics AEDFA-23-B-FA erbium-doped fiber amplifier (EDFA). The optical power launched into the \(5.2\)-km G.652.D fiber through the OC is approximately \(-9\)~dBm for both probes. A Halliburton PZ1-SMF4-APC-E PZT is connected approximately \(80\)~m from the fiber input and driven by a \(500\)-Hz sinusoidal voltage. Measurements are acquired with PZT drive voltages of \(5\) and \(1\)~V, representing strong- and weak-vibration conditions, respectively.

The Rayleigh-backscatter light returns through the OC, is filtered by the second EXFO XTM-50 OF, and is amplified by an Amonics AEDFA-PA-30-B-FA EDFA. In the LO branch, a polarization controller (PC) aligns the LO polarization with that of the returned signal. The two optical fields are mixed by an ULTRALOWNOISE BALANCEPHOTODETECTOR-300M-A balanced photodetector (BPD). For both probes, the amplified return and LO powers at the BPD inputs are approximately \(-14.5\)~dBm each. Its electrical output is sampled by a RIGOL DHO4404 digital storage oscilloscope (DSO) at \(1\)~GSa/s. Each \(50\)-million-sample record covers approximately \(50\)~ms, or about \(264\) probing periods.

For both probes, the acquired heterodyne BPD waveform is processed offline. A \(94\)--\(226\)-MHz radio-frequency bandpass filter first isolates the heterodyne signal band. After \(160\)-MHz digital downconversion, the signal is resampled to \(128\)~MSa/s and synchronized using the transmitted zero-guard energy edge, with affine timing alignment across probing periods. For CP-OFDM, the CP is removed before an \(N\)-point FFT. Let \(R_{\rm RF}[k]\) denote the normalized frequency response of the zero-phase digital RF bandpass filtering, evaluated at the RF frequency corresponding to subcarrier \(k\). The effective reference spectrum is given by \(S_{\rm eff}[k]=S[k]R_{\rm RF}[k]\). On the active subcarriers, the receiver uses~\eqref{eq:mmse_channel_est} with \(S[k]\) replaced by \(S_{\rm eff}[k]\) and \(\lambda=10^{-7}\max_{k\in\mathcal K_{\rm act}}|S_{\rm eff}[k]|^2\). The guard and DC estimates are set to zero before an \(8192\)-point inverse FFT reconstructs the band-limited Rayleigh-backscatter channel.  
For LFM probing, the same digital RF filter response is incorporated into the transmitted reference used for energy-normalized matched filtering in~\eqref{eq:mf_channel_est}. 
For noiseless ZF reconstruction, the active-frequency mask yields the spatial impulse response \(a_N[d]=N^{-1}\sum_{k\in\mathcal K_{\rm act}}e^{j2\pi kd/N}\), where \(d\) is an integer range-bin offset. The DC null makes the aperture noncontiguous, and thus this response differs from the contiguous-aperture expression in~\eqref{eq:ofdm_dirichlet_res}.  

To mitigate coherent Rayleigh fading, both receivers use the common channel-domain combining in~\eqref{eq:fixed_channel_combining} before gauge-differential phase recovery. The frequency views comprise the full occupied band and eight overlapping Hann subbands. Each subband has a support width of \(B_{\rm occ}/2\) and an adjacent shift of \(B_{\rm occ}/14\). For either reconstructed channel, the \(q\)-th view is
\begin{equation}
    \hat h_{p,m}^{(q)}
    =\operatorname{IDFT}\!\left\{
    B_q[k]\operatorname{DFT}\{\hat h_p\}[k]
    \right\}[m],
    \label{eq:exp_subband_gauge}
\end{equation}
where \(B_q[k]\) is the frequency mask and \(\hat h_p\) denotes the full reconstructed range-domain sequence before range selection. The transforms use the reconstruction grid of the corresponding receiver, and the full-band mask retains the occupied frequency samples. The fixed complex coefficients account for the correlated noise among overlapping views using a receiver-specific noise covariance model. These coefficients are not updated using the instantaneous branch amplitudes.

Reliable phase observations require both gauge endpoints to satisfy $ \min_{m\in\{\ell,\ell+G\}}\frac{|\bar h_{p,m}|}{\sigma_{p,m}}\ge 2$, where \(\sigma_{p,m}^2\) is the combined complex-noise variance estimated by propagating the noise measured in the pulse-free interval through the receiver. Near-\(\pi\) phase transitions are additionally rejected using a three-standard-deviation uncertainty margin. The retained observations are unwrapped separately on each reliable interval using~\eqref{eq:common_phase_recovery}, followed by zero-phase forward--backward filtering with a fourth-order Butterworth low-pass filter at \(1400\)~Hz, linear detrending, and temporal-mean removal within that interval. No filtering or phase connection is performed across missing observations. A \(13\)-sample gauge separation on the nominal \(0.78125\)-m grid gives a nominal gauge length of \(10.15625\)~m.
In addition, the blind frequency search is restricted to \(0<f\le1400\)~Hz.
 
Figs.~\ref{Resfig1VDP} and~\ref{Resfig5VDP} compare the differential-phase traces over all probing periods for the proposed CP-OFDM and the traditional matched-filter-based LFM (MF-LFM) receivers under \(1\)- and \(5\)-V, \(500\)-Hz PZT drives at about \(80\)~m from the fiber input, respectively. 
The root-mean-square (RMS) differential phase is calculated over valid spatial and temporal observations within the selected approximately \(10\)-m vibration window and the remaining background noise region, yielding the vibration RMS (VRMS) and background RMS (BRMS), respectively.
The VBR is defined as \(\mathrm{VBR}=20\log_{10}(\mathrm{VRMS}/\mathrm{BRMS})\).
At \(1\)~V, CP-OFDM and MF-LFM yield VRMS of \(0.5267\) and \(0.5141\)~rad and BRMS  of \(0.06664\) and \(0.09395\)~rad, corresponding to VBR of \(17.96\) and \(14.76\)~dB, respectively.
At \(5\)~V, both VRMS are  \(2.5\)~rad, while the BRMS are \(0.1088\) and \(0.2209\)~rad and the VBR are \(27.23\) and \(21.07\)~dB, respectively. 
The nearly fivefold increase in VRMS is consistent with the proportional relationship between the recovered vibration amplitude and the PZT drive voltage at the two tested levels.

\begin{figure}[!t]
    \centering
    \includegraphics[scale=0.082]{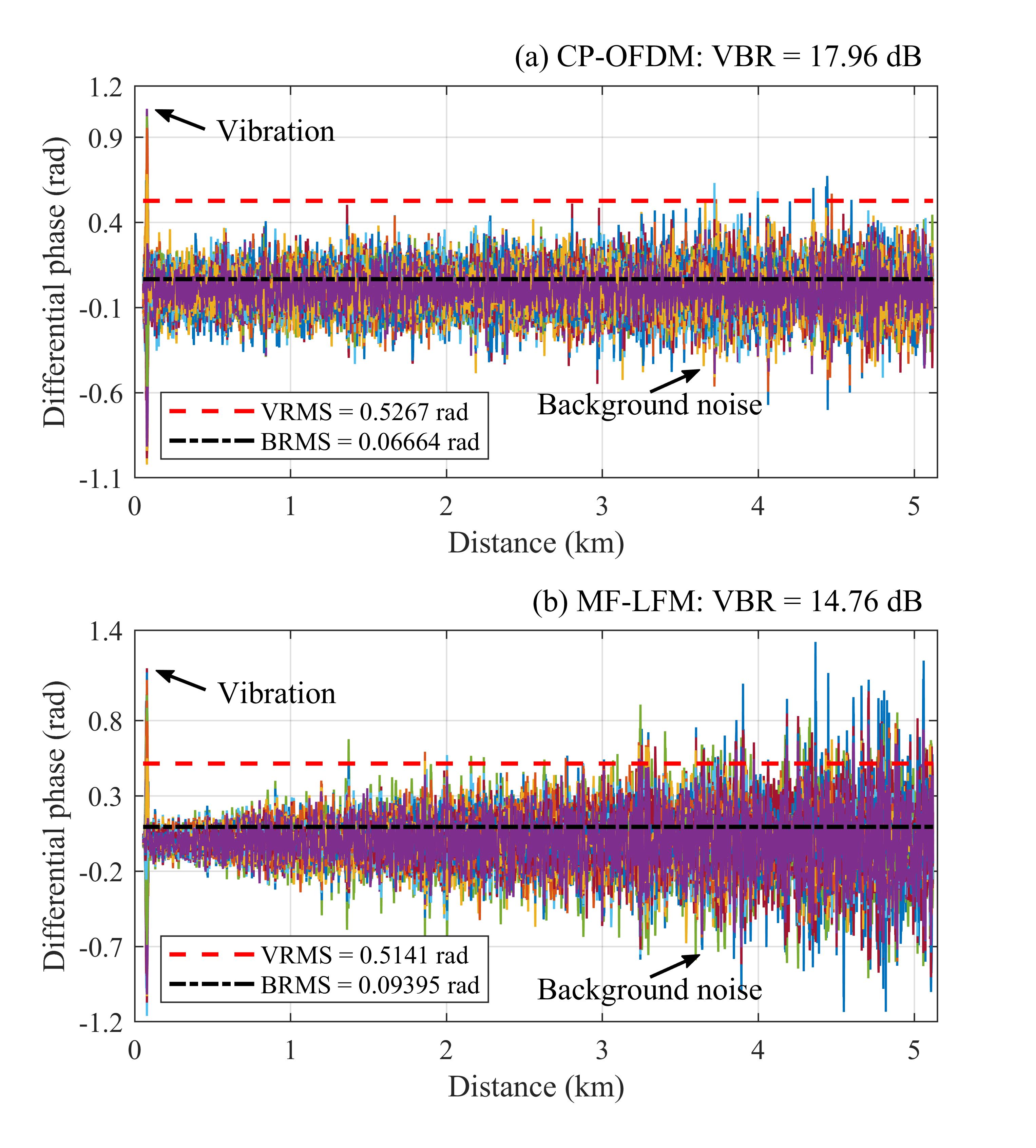}
    \caption{Experimental differential-phase traces over all probing periods under a 1-V, 500-Hz PZT drive at about \(80\)~m from the fiber input, where the root-mean-square (RMS) differential phase over observations is \(0.5267\) and \(0.5141\)~rad within the selected about \(10\)-m vibration window near \(80\)~m (vibration RMS, VRMS) and \(0.06664\) and \(0.09395\)~rad over the remaining background noise region (background RMS, BRMS), yielding vibration-to-background ratios \(\mathrm{VBR}=20\log_{10}(\mathrm{VRMS}/\mathrm{BRMS})\) of \(17.96\) and \(14.76\)~dB for (a) the proposed CP-OFDM and (b) the matched-filter-based LFM (MF-LFM) receivers, respectively.}
    \label{Resfig1VDP}
\end{figure}

Using the same vibration site, occupied bandwidth, emitted pulse duration, repetition period, launch power, and common subsequent phase processing, CP-OFDM improves the VBR by approximately \(3.2\) and \(6.2\)~dB at \(1\) and \(5\)~V, respectively.
The similar VRMS values and lower BRMS  values indicate that the improvement primarily results from reduced background phase fluctuations while maintaining the vibration response.
This result is consistent with the suppression of matched-filter-induced spatial ISI analyzed in Section~\ref{ISIByMF}, although these single-PZT measurements do not isolate spatial ISI from phase noise, receiver noise, and residual fading.

\begin{figure}[!t]
    \centering
    \includegraphics[scale=0.082]{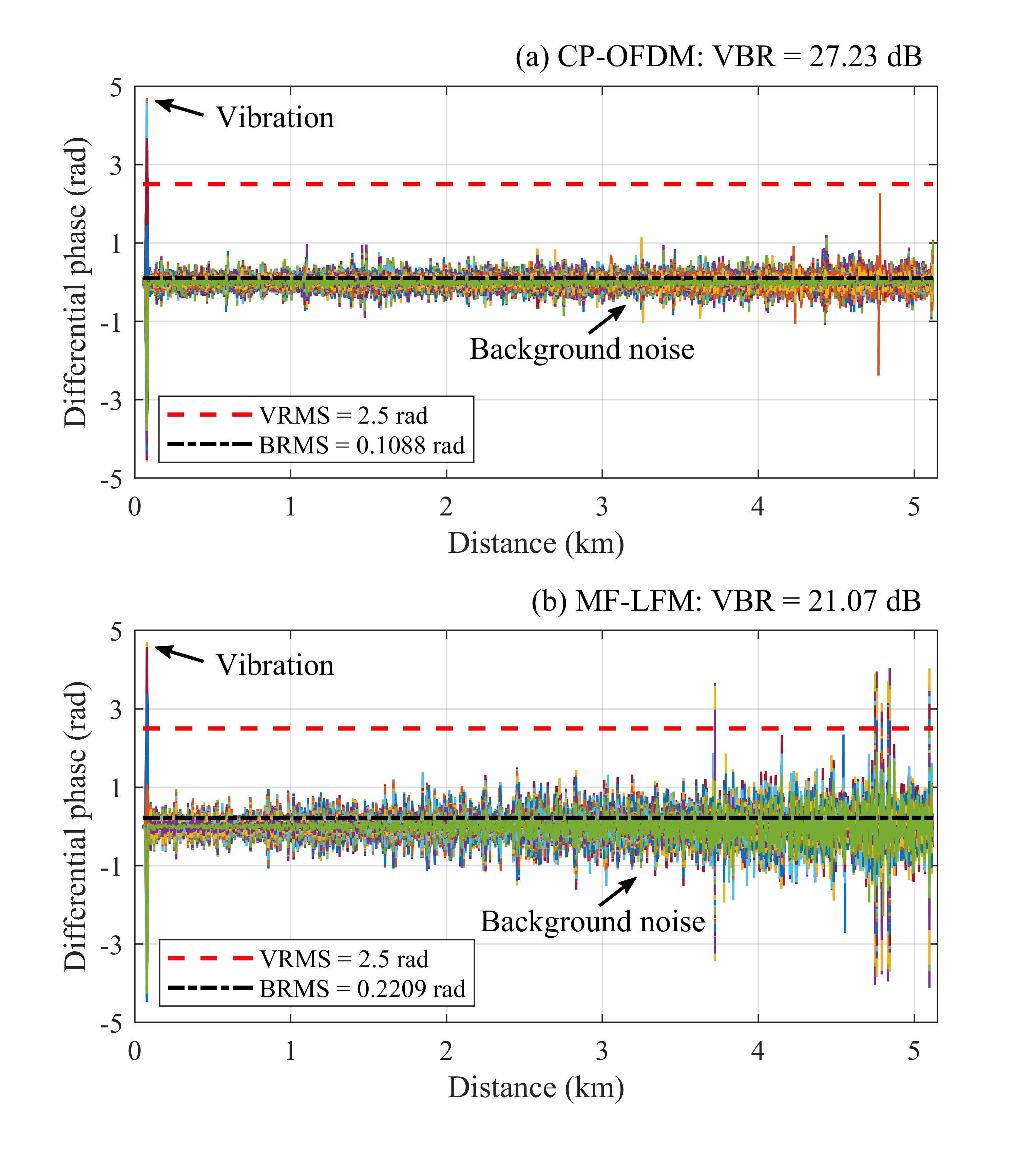}
    \caption{Experimental differential-phase traces over all probing periods under a 5-V, 500-Hz PZT drive at about \(80\)~m from the fiber input, where the root-mean-square (RMS) differential phase over observations is \(2.5\) and \(2.5\)~rad within the selected about \(10\)-m vibration window near \(80\)~m (vibration RMS, VRMS) and \(0.1088\) and \(0.2209\)~rad over the remaining background noise region (background RMS, BRMS), yielding vibration-to-background ratios \(\mathrm{VBR}=20\log_{10}(\mathrm{VRMS}/\mathrm{BRMS})\) of \(27.23\) and \(21.07\)~dB for (a) the proposed CP-OFDM and (b) the matched-filter-based LFM (MF-LFM) receivers, respectively.}
    \label{Resfig5VDP}
\end{figure}

The MF-LFM traces also exhibit a more pronounced increase in background differential-phase fluctuations toward the far end of the fiber.
Laser phase noise disrupts the coherent accumulation in matched filtering, degrading the pulse-compression response~\cite{LoayssaJLT2022PhaseNoise}. Longer LFM pulses can exacerbate this impairment and limit the benefit of additional probe energy~\cite{PineiroJLT2023}. This mechanism provides a plausible contribution to the distance dependence observed in~Figs.~\ref{Resfig1VDP} and~\ref{Resfig5VDP}, alongside fiber attenuation and residual fading.
The CP-OFDM receiver uses frequency-domain channel reconstruction without matched-filter pulse compression, resulting in a weaker background increase in these measurements.
These results support the use of channel reconstruction with long probes and motivate continuous payload-bearing operation~\cite{HuangDMTISAC2026}, subject to adequate channel coherence over each reconstruction window.

\begin{figure}[!t]
    \centering
    \includegraphics[scale=0.082]{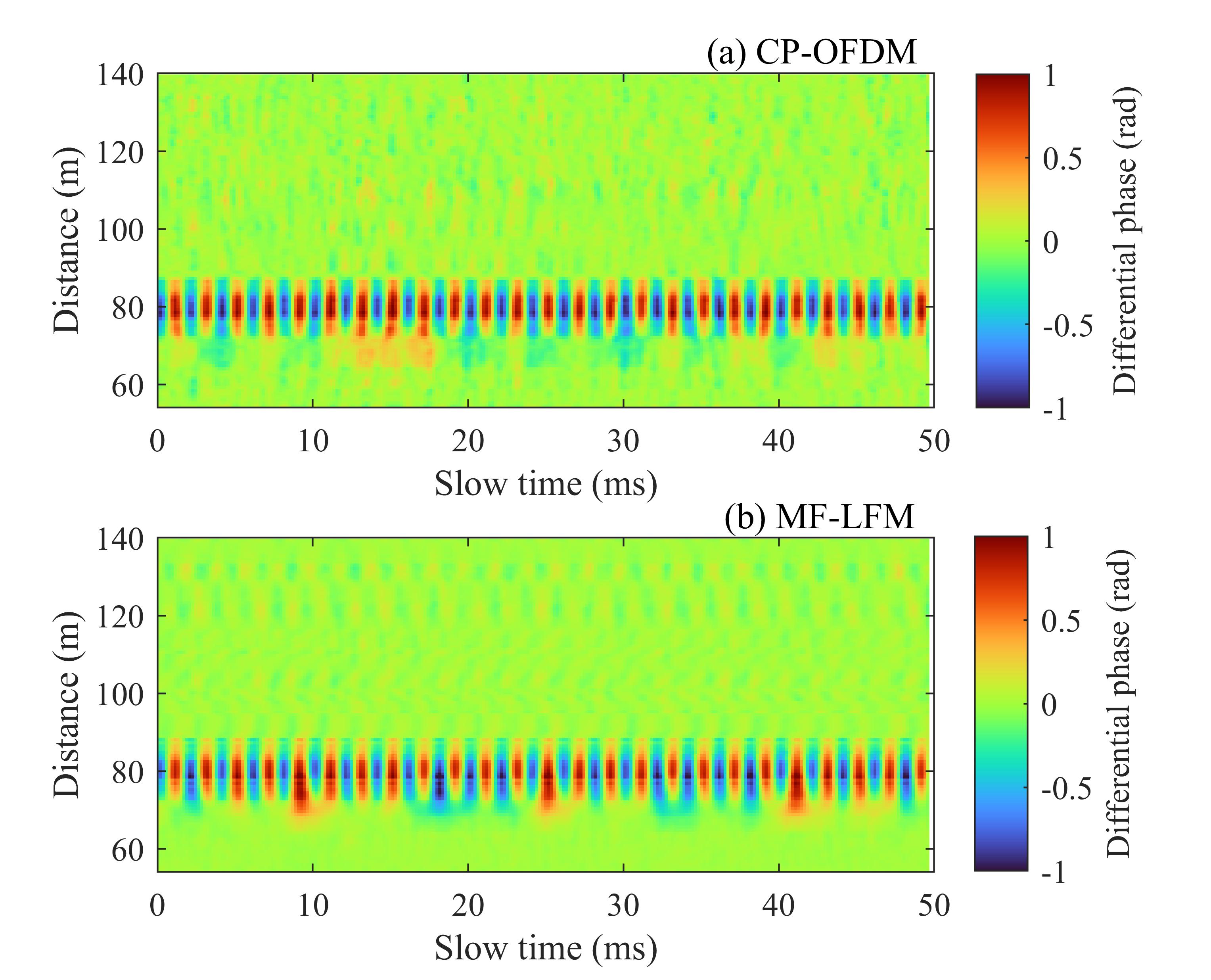}
    \caption{Experimental slow-time variation of the recovered differential phase around the vibration region near \(80\)~m under a 1-V, 500-Hz PZT drive for (a) the proposed CP-OFDM and (b) the matched-filter-based LFM (MF-LFM) receivers, where both panels use the same color scale from \(-1\) to \(1\)~rad.}
    \label{Resfig1VDP2D}
\end{figure}

Figs.~\ref{Resfig1VDP} and~\ref{Resfig5VDP} superimpose the spatial differential-phase traces from all probing periods, showing the range distribution of the vibration response and background fluctuations.
To examine temporal stability, Figs.~\ref{Resfig1VDP2D} and~\ref{Resfig5VDP2D} plot the same observations around the vibration region along the slow-time axis.
For both drive levels, the response remains localized near \(80\)~m throughout the approximately \(50\)-ms record, while the alternating positive and negative phase pattern repeats with a period of approximately \(2\)~ms, consistent with the \(500\)-Hz excitation.

\begin{figure}[!t]
    \centering
    \includegraphics[scale=0.082]{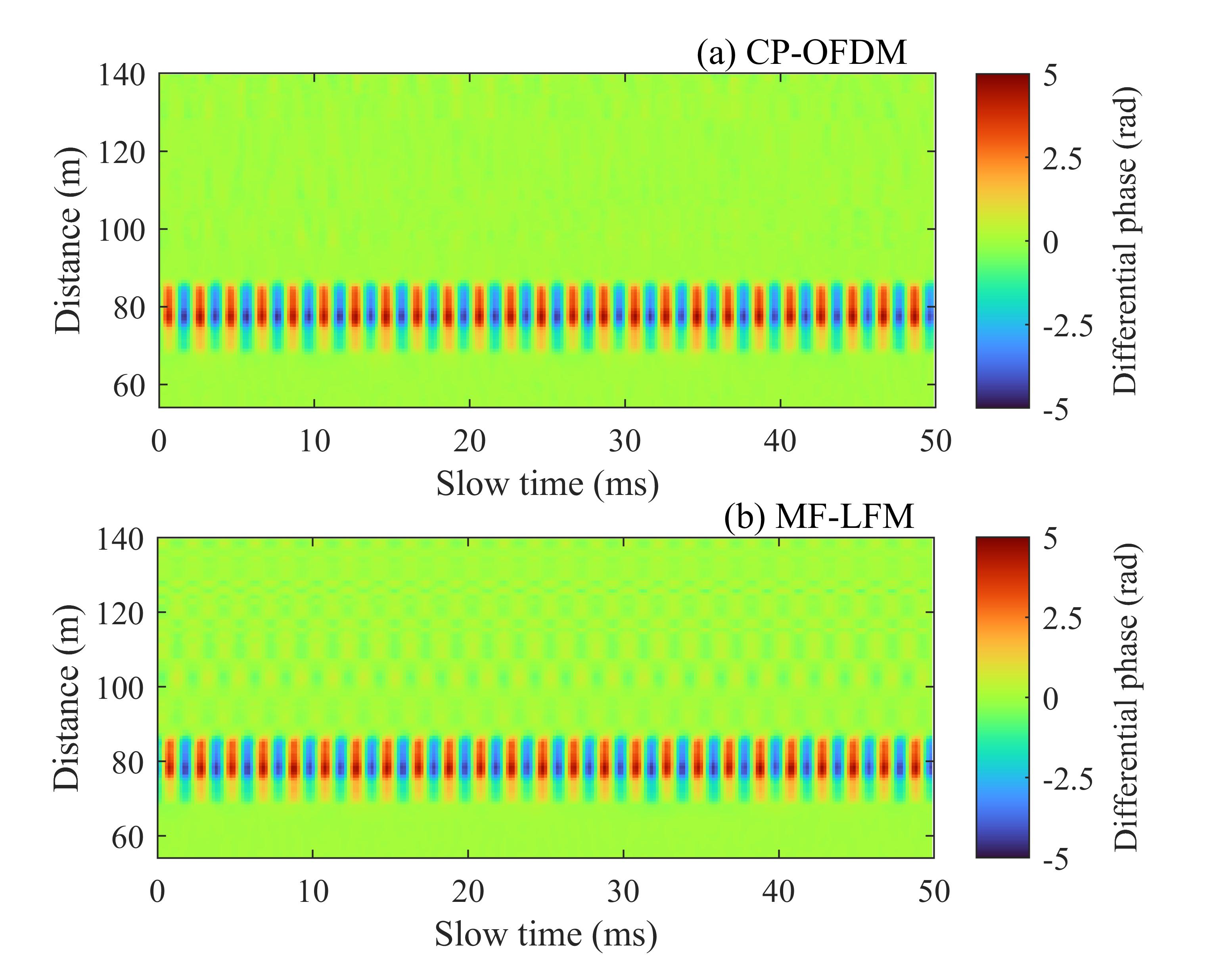}
        \caption{Experimental slow-time variation of the recovered differential phase around the vibration region near \(80\)~m under a 5-V, 500-Hz PZT drive for (a) the proposed CP-OFDM and (b) the matched-filter-based LFM (MF-LFM) receivers, where both panels use the same color scale from \(-5\) to \(5\)~rad.}
    \label{Resfig5VDP2D}
\end{figure}

At \(1\)~V, both the proposed CP-OFDM and the MF-LFM receivers exhibit identifiable periodic responses despite fluctuations in the surrounding range bins.
At \(5\)~V, the larger differential-phase excursion preserves the spatial localization and temporal periodicity. The CP-OFDM map shows less pronounced fluctuations outside the vibration region than the MF-LFM map.
The color limits are \(\pm1\) and \(\pm5\)~rad for the \(1\)- and \(5\)-V cases, respectively, with a common scale for the two receivers within each figure.
Figs.~\ref{Resfig1VDP2D} and~\ref{Resfig5VDP2D} supplement the aggregate VRMS and BRMS measurements by displaying sustained, spatially localized phase oscillations over successive probing periods.

\begin{figure}[!t]
    \centering
    \includegraphics[scale=0.082]{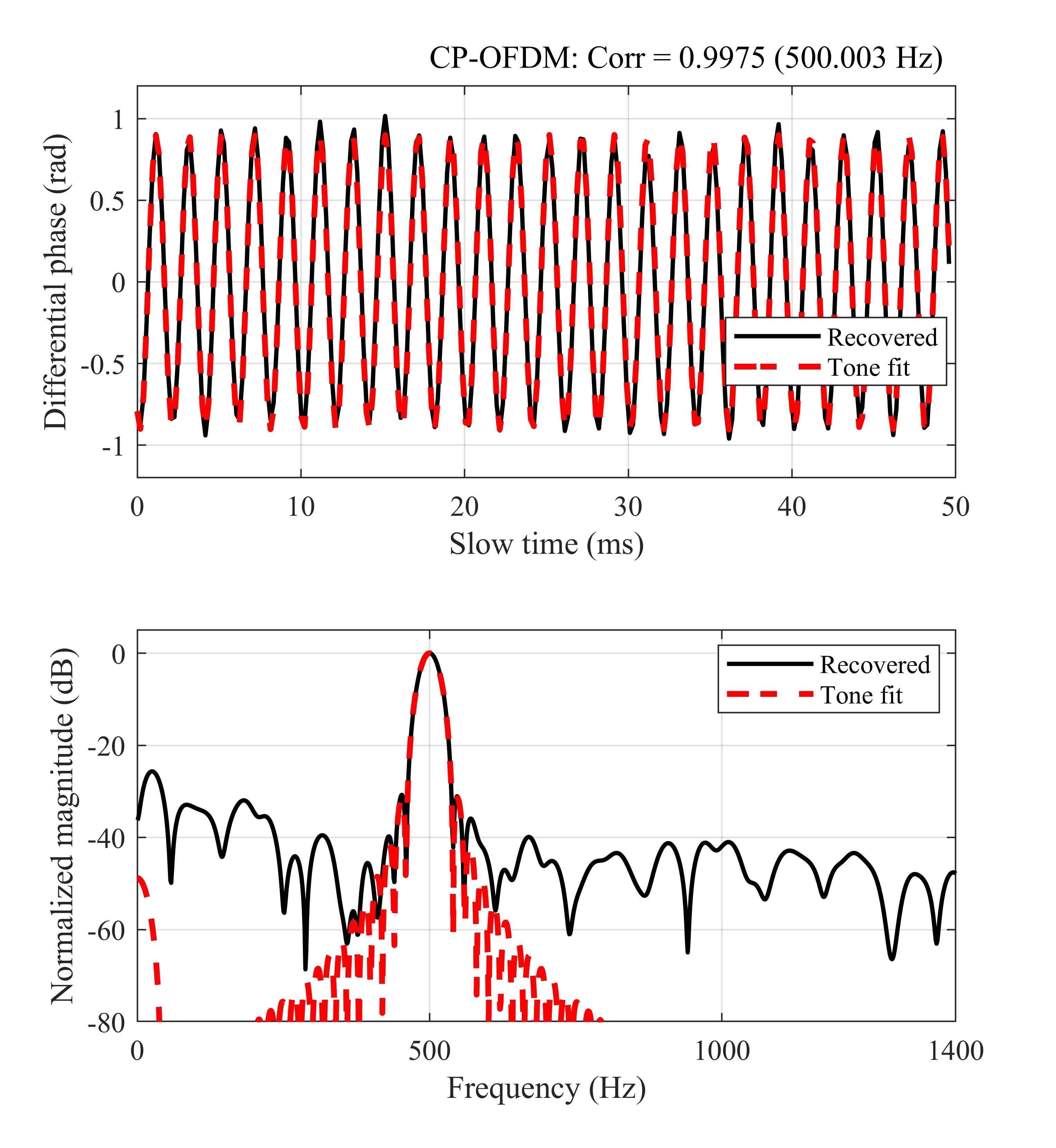}
    \caption{Experimental vibration waveform and corresponding spectrum recovered by the proposed CP-OFDM receiver under a 1-V, 500-Hz PZT drive, where the upper panel compares the recovered differential-phase trace with a fitted sinusoidal tone, yielding a correlation coefficient of \(0.9975\) and a fitted frequency of \(500.003\)~Hz, and the lower panel compares their magnitude spectra normalized to the peak magnitude of the recovered-trace spectrum.}
    \label{Resfig1VDPAF_OFDM}
\end{figure}

To evaluate the waveform recovery associated with these periodic phase patterns, Figs.~\ref{Resfig1VDPAF_OFDM}--\ref{Resfig5VDPAF_LFM} display the differential-phase traces at the center range bins of the chosen vibration windows, along with their respective spectra.
Each sinusoidal tone is fitted to the recovered trace, and its frequency is estimated from the data without being fixed at \(500\)~Hz.
The spectra exhibit a dominant component near \(500\)~Hz. Within each figure, the recovered and fitted spectra are normalized to the peak magnitude of the recovered-trace spectrum.

\begin{figure}[!t]
    \centering
    \includegraphics[scale=0.082]{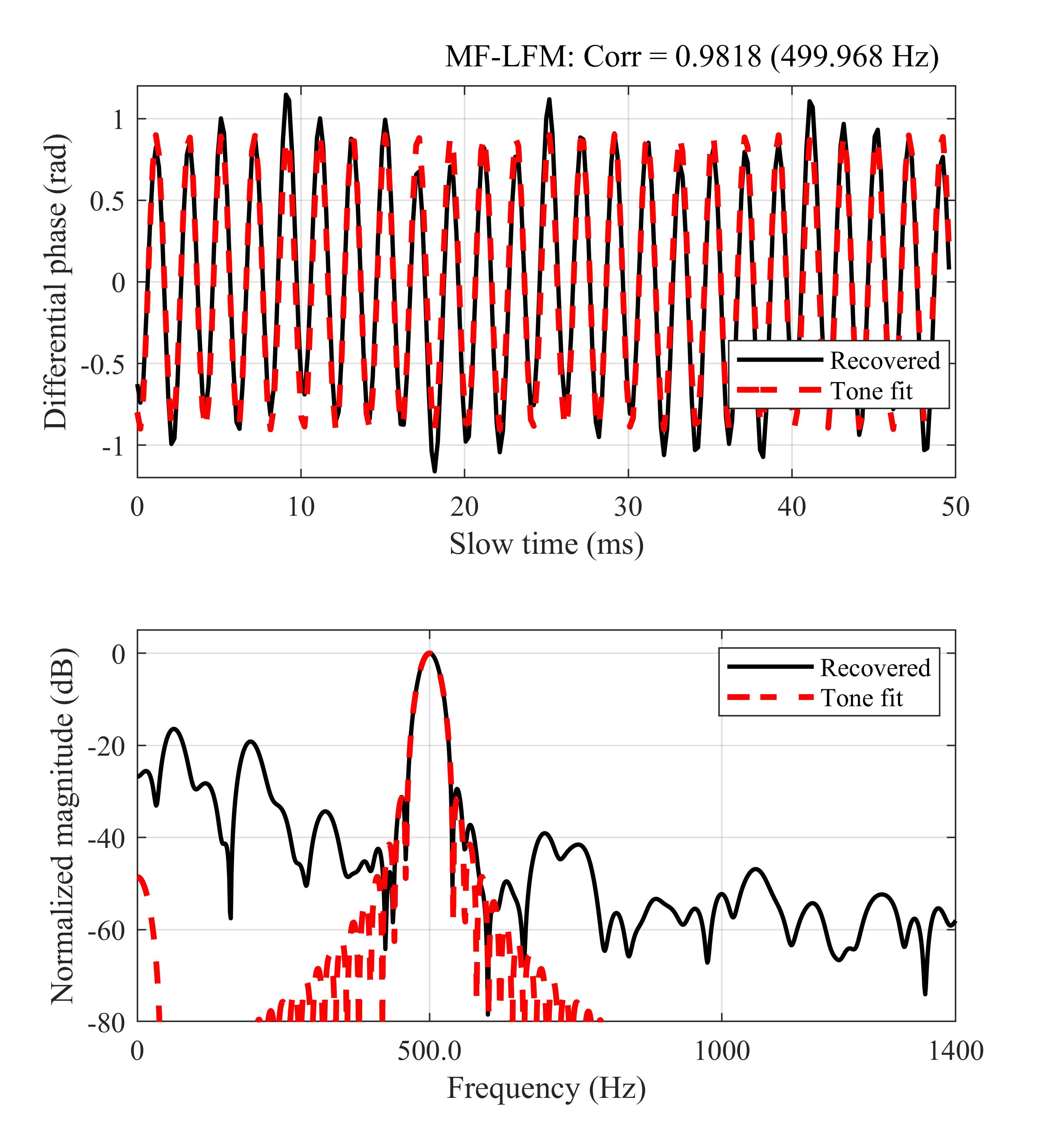}
        \caption{Experimental vibration waveform and corresponding spectrum recovered by the  matched-filter-based LFM (MF-LFM) receiver under a 1-V, 500-Hz PZT drive, where the upper panel compares the recovered differential-phase trace with a fitted sinusoidal tone, yielding a correlation coefficient of \(0.9818\) and a fitted frequency of \(499.968\)~Hz, and the lower panel compares their magnitude spectra normalized to the peak magnitude of the recovered-trace spectrum.}
    \label{Resfig1VDPAF_LFM}
\end{figure}

\begin{figure}[!t]
    \centering
    \includegraphics[scale=0.082]{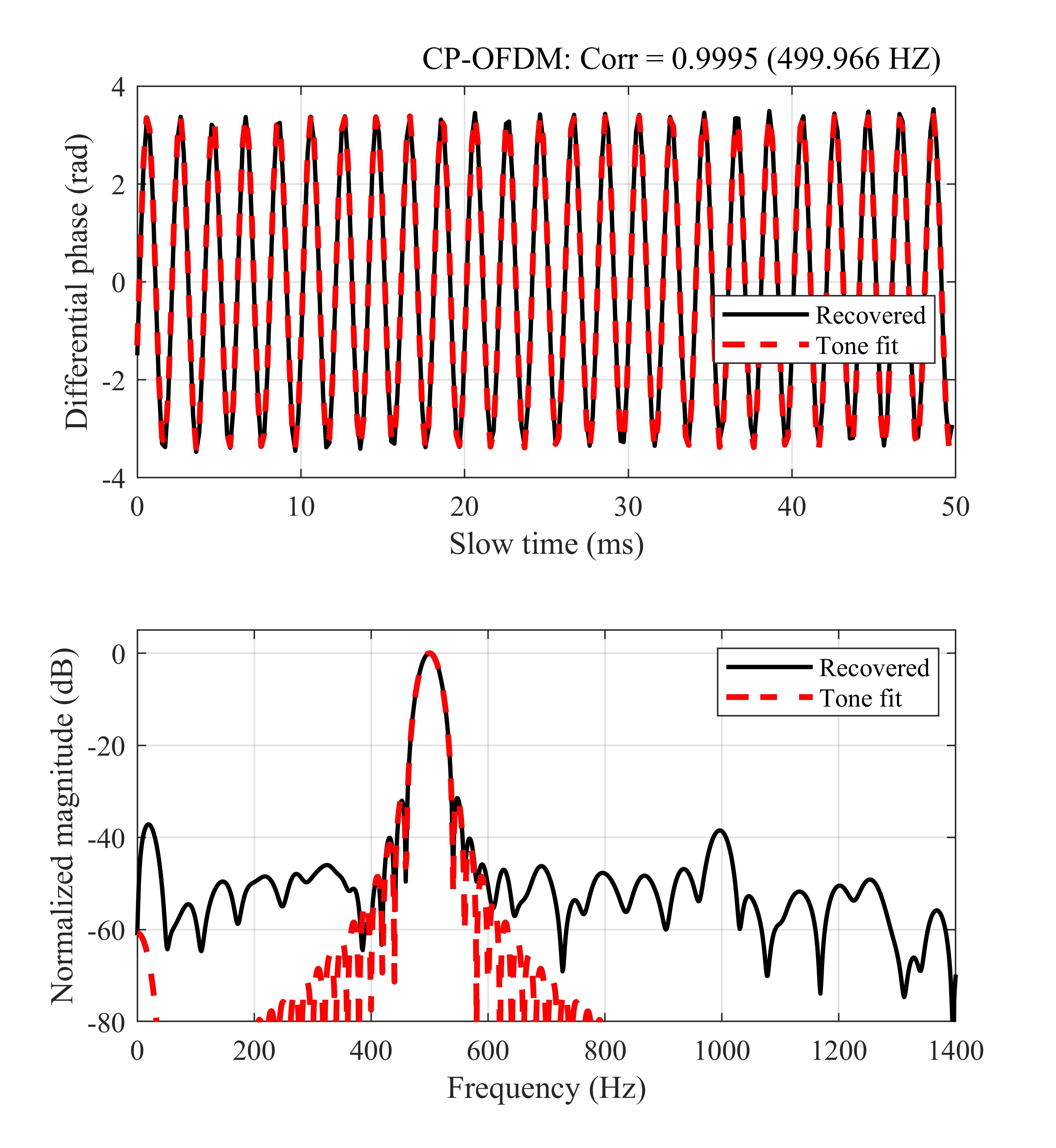}
    \caption{Experimental vibration waveform and corresponding spectrum recovered by the proposed CP-OFDM receiver under a 5-V, 500-Hz PZT drive, where the upper panel compares the recovered differential-phase trace with a fitted sinusoidal tone, yielding a correlation coefficient of \(0.9995\) and a fitted frequency of \(499.966\)~Hz, and the lower panel compares their magnitude spectra normalized to the peak magnitude of the recovered-trace spectrum.}
    \label{Resfig5VDPAF_OFDM}
\end{figure}

At \(1\)~V, the proposed CP-OFDM and the traditional MF-LFM receivers yield correlation coefficients of \(0.9975\) and \(0.9818\), respectively, with fitted frequencies of \(500.003\) and \(499.968\)~Hz. This indicates closer agreement with a sinusoidal waveform for CP-OFDM.
At \(5\)~V, the corresponding correlation coefficients increase to \(0.9995\) and \(0.9999\), with fitted frequencies of \(499.966\) and \(500.007\)~Hz, respectively.
Therefore, both recover a nearly sinusoidal response at the higher drive level, and all four fitted frequencies are consistent with the nominal \(500\)-Hz excitation.

\begin{figure}[!t]
    \centering
    \includegraphics[scale=0.082]{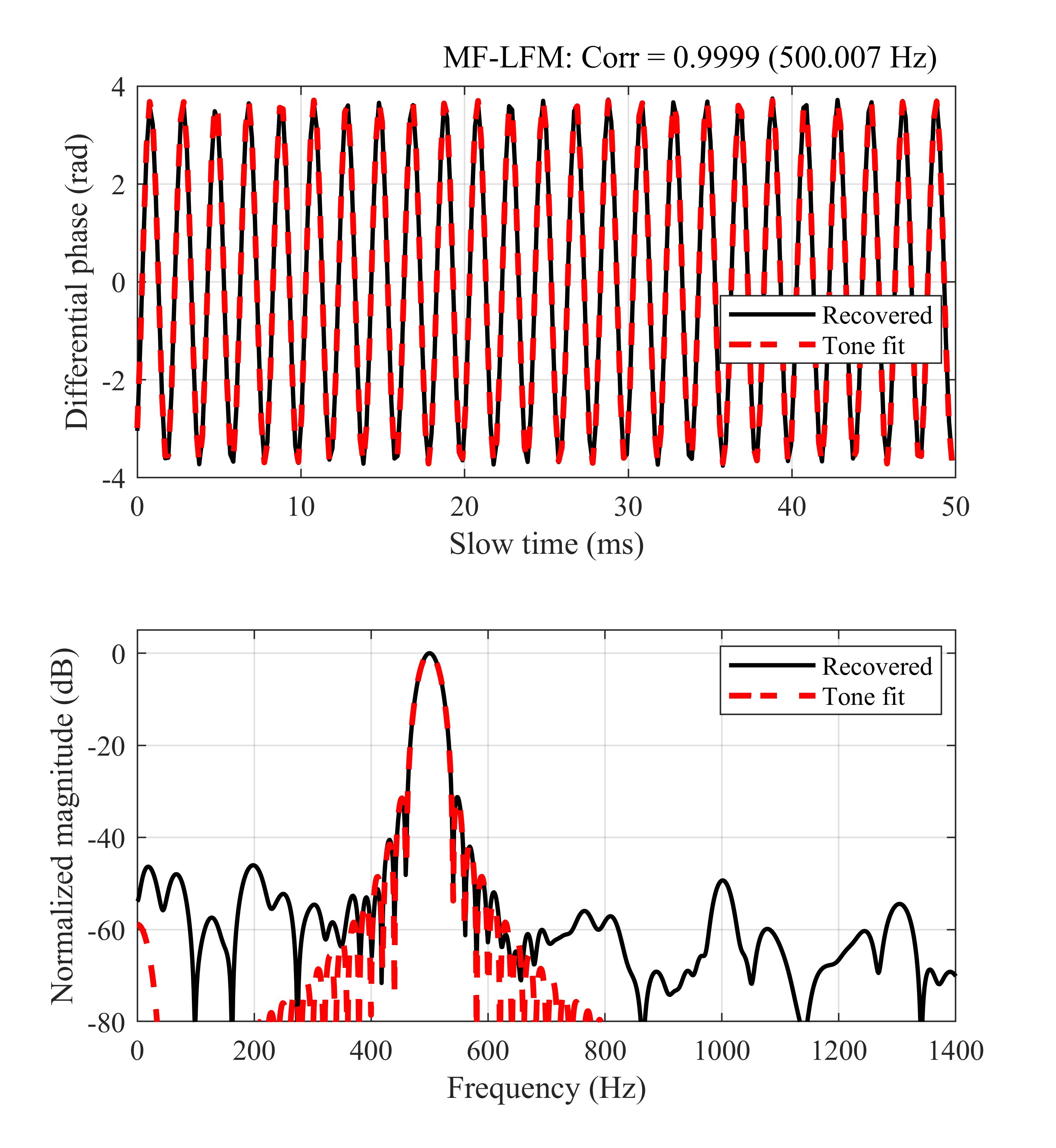}
        \caption{Experimental vibration waveform and corresponding spectrum recovered by the  matched-filter-based LFM (MF-LFM) receiver under a 5-V, 500-Hz PZT drive, where the upper panel compares the recovered differential-phase trace with a fitted sinusoidal tone, yielding a correlation coefficient of \(0.9999\) and a fitted frequency of \(500.007\)~Hz, and the lower panel compares their magnitude spectra normalized to the peak magnitude of the recovered-trace spectrum.}
    \label{Resfig5VDPAF_LFM}
\end{figure}

The correlations quantify the agreement of the recovered data with a sinusoid and complement the spatial VBR measurements.
The slightly higher MF-LFM correlation at \(5\)~V is consistent with the higher CP-OFDM VBR. The MF-LFM correlation characterizes the waveform at one range bin, while the CP-OFDM VBR measures the vibration response relative to fluctuations in the background region.
Together, these observations demonstrate that CP-OFDM preserves the local vibration waveform while reducing the spatial background of the tested measurements.

\section{Conclusions}\label{Conclu}
We developed a CP-OFDM DAS system for spatial-ISI-free $\phi$-OTDR using a data-bearing waveform as the sensing probe. By replacing matched-filter-based pulse compression with CP-aided frequency-domain channel reconstruction, the proposed CP-OFDM receiver directly recovers the distributed Rayleigh-backscatter channel before gauge-differential phase demodulation. The same probe also supported forward communication data recovery, providing an initial shared-waveform optical-fiber ISAC demonstration.

For the first time, distributed Rayleigh backscattering was formulated as a finite-memory sensing multipath channel. This mathematical equivalence shows that conventional pulse-compression receivers reconstruct the channel convolved with the probing-waveform autocorrelation, whose nonzero sidelobes produce deterministic waveform-induced spatial ISI. The dense-event simulations verified this mechanism and showed improvements in phase-trace mean-square error of up to 29.55~dB over the LFM matched-filter receiver. 
We further proved that, when the useful OFDM length and CP length cover the sensing multipath channel memory, the reconstructed range-bin response contains the desired channel coefficient and noise but no deterministic leakage from other range bins, thereby establishing spatial-ISI-free channel reconstruction on the modeled range-bin grid. 
This property does not imply super-resolution: CP-OFDM and LFM probing retain the same bandwidth-limited spatial resolving interval under the same occupied bandwidth. Experiments over a 5.2-km fiber link showed that CP-OFDM preserved the local 500-Hz vibration waveform while improving the VBR by approximately 3.2 and 6.2~dB over matched-filter-based LFM pulse compression under 1- and 5-V PZT drives, respectively.

The present proof-of-concept uses periodically launched CP-OFDM probing blocks rather than a continuous communication waveform. Future work will investigate continuous CP-OFDM transmission with joint synchronization, channel tracking, and sensing processing to support uninterrupted communication and distributed sensing, thereby advancing toward a practical optical-fiber ISAC system.
 

\end{document}